\documentclass[11pt]{article}
\usepackage{amsmath,amsthm,amssymb}
\usepackage{fullpage}

\usepackage[shortlabels]{enumitem}
\setitemize{noitemsep,topsep=0pt,parsep=0pt,partopsep=0pt}

\usepackage{url}

\usepackage{bm}

\usepackage{tabulary}
\usepackage{multicol}
\usepackage{caption,subcaption}
\usepackage{tikz}
\usetikzlibrary{arrows.meta}
\tikzset{>={Latex[width=2mm,length=2mm]}}

\usetikzlibrary{decorations.pathreplacing}

\tikzstyle{vertex}=[circle, draw, inner sep=0pt, minimum size=6pt]

\newcommand{\IGNORE}[1]{}
\newcommand{\noi}{{\noindent}}

\newcommand{\wt}{\ensuremath{\widetilde}}

\newcommand{\opt}{\textsl{opt}}

\newtheorem{theorem}{Theorem}[section]
\newtheorem{proposition}[theorem]{Proposition}
\newtheorem{lemma}[theorem]{Lemma}

\newtheorem{claim}[theorem]{Claim}
\newtheorem{defn}{Definition}[section]

\newcommand{\authremark}[1]{\smallskip\noi\textbf{Remark}:{#1}\newline\smallskip}

\newcommand{\DTWO}{{\textup{D2}\/}}

\newcommand{\cost}{\textup{cost}}
\newcommand{\credit}{\textup{credit}}

\newcommand{\bcredit}{\textup{b-credit}}
\newcommand{\ccredit}{\textup{c-credit}}
\newcommand{\ncredit}{\textup{n-credit}}

\bmdefine{\bstar}{\star}
\newcommand\wsMAP{MAP$\bstar$}

\newcommand{\zsplit}{\textup{zero-cost~S2}\/}
\newcommand{\usplit}{\textup{unit-cost~S2}\/}
\newcommand{\csplit}{\textup{S$\{3,4\}$}\/}
\newcommand{\redcycle}{\textup{R4}\/}
\newcommand{\redgadget}{\textup{R8}\/}

\newcommand{\wth}{\widetilde{H}}
\newcommand{\wtg}{\widetilde{G}}
\newcommand{\wtc}{\widetilde{C}}
\newcommand{\wtp}{\widetilde{P}}
\newcommand{\wtr}{\widetilde{R}}
\newcommand{\wtv}{\widetilde{v}}
\newcommand\augC{\hat{Q}}

\newcommand{\sba}{\ensuremath{\mathcal{A}}}

\newcommand{\att}{\ensuremath{\mathcal{T}}}

\newcommand{\Nmin}{12}

\newcommand{\hv}{\hat{v}}
\newcommand{\hB}{\hat{B}}

\begin{document}

\title
{
An Improved Approximation Algorithm for
the Matching Augmentation Problem
}
\author{
J.Cheriyan	\thanks{C\&O~Dept., University of Waterloo, Canada}
\and
R.Cummings	\thanks{C\&O~Dept., University of Waterloo, Canada}
\and
J.Dippel	\thanks{McGill University, Montreal, Canada}
\and
J.Zhu		\thanks{C\&O~Dept., University of Waterloo, Canada}
}

  \date{\today}

\maketitle

\vspace{-0.25in}

\begin{abstract}
We present a $\frac53$-approximation algorithm for the matching augmentation
problem (MAP): given a multi-graph with edges of cost either zero or one
such that the edges of cost~zero form a matching, find a 2-edge connected
spanning subgraph (2-ECSS) of minimum cost.

A $\frac74$-approximation algorithm for the same problem was presented
recently, see Cheriyan, et~al.,
``The matching augmentation problem: a $\frac{7}{4}$-approximation
algorithm," {\em Math. Program.}, 182(1):315--354, 2020.

Our improvement is based on new algorithmic techniques, and some
of these may lead to advances on related problems.
\medskip

\end{abstract}

\clearpage

\section{ \label{s:intro}  Introduction}

The design and analysis of algorithms for problems in network design is
a core topic in Theoretical Computer Science and
Combinatorial Optimization.  Algorithmic research on problems such
as the minimum spanning tree problem and the Traveling Salesman
Problem (TSP) started decades ago and is a thriving area even today.
One of the key problems in this area is the minimum-cost 2-ECSS (2-edge
connected spanning subgraph) problem:
Given an undirected graph $G=(V,E)$ and a nonnegative cost for each
edge $e\in{E}$, denoted $\cost(e)$, find a minimum-cost spanning
subgraph $H=(V,F),\; F\subseteq{E}$, that is 2-edge connected.
Throughout, we use $n:=|V|$ to denote the number of nodes of $G$.
(Recall that a graph is 2-edge connected if it is connected and has
no ``cut edges", or equivalently, each of its nontrivial cuts has 
$\geq2$ edges.) This problem is NP-hard, and the best approximation
guarantee known, due to \cite{KV94}, is~2.

On the other hand, the best ``hardness of approximation threshold"
known is much smaller; for example, it is $(1 + \frac{\rho_{VC3}}{104})$
for the unweighted problem, where $1+\rho_{VC3}$ is the ``hardness
of approximation threshold" for the minimum vertex cover problem
on a graph with maximum degree~3, \cite[Theorem~5.2]{GGTW09}.
Also, the best lower~bound known on the integrality ratio of the
standard LP~relaxation (for minimum-cost 2-ECSS) is around 1.5 (thus,
well below~2), see \cite{CKKK08}.

\subsection{FAP, TAP and MAP}

Given this significant gap between the lower~bounds and the
upper~bounds, research in this area over the last two decades has
focused on the case of zero-one cost functions (every edge has a
cost of zero or one).
Let us call an edge $e\in E$ with $\cost(e)=0$ a zero-edge, and
let us call an edge $e\in E$ with $\cost(e)=1$ a unit-edge.
Intuitively, the zero-edges define some existing network that we
wish to augment (with unit-edges) such that the augmented
network is resilient to the failure of any one edge.
We may assume that the zero-edges form a forest; otherwise, there
is at least one cycle $C$ formed by the zero~edges, and in that
case, we may contract $C$, solve the problem on the resulting graph
$G/C$, find a solution (edge~set) $F$, and return
$F\cup{C}$ as a solution of the original problem.
Consequently, the minimum-cost 2-ECSS problem with a zero-one cost
function is called the \textit{Forest Augmentation Problem} or FAP.
The challenge is to design an approximation algorithm with guarantee
strictly less than~2 for FAP.

A well~known special case of FAP is TAP, the \textit{Tree Augmentation
Problem}: the set of zero-edges forms a spanning tree.
The first publication to break the ``2-approximation barrier" for
TAP is \cite{Na03} (2003), and since then there have been several
important advances, including recent work, see
\cite{EFKN09,KN:talg16,Adj17,N:esa17,CG:tap18,FGKS:soda18,GKZ:stoc18}.
Starting with the results of \cite{Adj17} (2017),
the improved approximation guarantees hold also
for a weighted version of TAP
where the edge-costs are bounded by a constant, that is,
the edge-costs are in the interval $[1,M]$, where $M=O(1)$.

Recently, see \cite{cdgkn:map}, there has been progress on another
important (in our opinion) special case of FAP called the
\textit{Matching Augmentation Problem} or MAP:
Given a multi-graph with edges of cost either zero or one such that
the zero-edges form a matching, find a 2-ECSS of minimum cost.
From the view-point of approximation algorithms,
MAP is ``complementary'' to TAP,
in the sense that the forest formed on $V(G)$ by the zero-edges has
many connected components, each with one node or two nodes,
whereas this forest has only one connected~component in TAP.

\subsection{Previous literature and possible approaches for attacking MAP}

Given the large body of work on network design and the design of
algorithms (for finding optimal solutions, as well as for finding
approximately optimal solutions), see the books in the area
\cite{Schrijver,WS:book,LRS:book}, one would expect some way of
breaking the ``2-approximation barrier" for FAP.  Unfortunately,
no such method is known (to the best of our knowledge).

Powerful and versatile methods such as the \textit{primal-dual
method} (see \cite{WS:book,GW95}) and the \textit{iterative rounding
method} (see \cite{LRS:book,J01}) have been developed for problems
in network design, but the proveable approximation guarantees for
these methods are $\ge2$.
\big(These methods work by rounding LP~relaxations, and informally
speaking, the approximation guarantee is proved via an upper bound
of~2 per~iteration on the ``integral cost incurred" versus the
``chargeable LP cost", and it is plausible that the factor of~2
cannot be improved for this type of analysis.\big)

Another important sequence of recent advances due to
\cite{Adj17,N:esa17,FGKS:soda18,GKZ:stoc18}
proves approximation guarantees (well) below~2 for TAP, based on a
new family of LP~relaxations that have so-called \textit{bundle
constraints}; these constraints are defined by a set of paths of
zero-edges.
These methods rely on the fact that the set of zero-edges forms
a connected graph that spans all the nodes, and unfortunately, this
property does not hold for MAP.

Combinatorial methods that may also exploit lower-bounds from
LP~relaxations have been developed for approximation algorithms for
unweighted minimum-cost 2-ECSS,
e.g., $\frac43$-approximation algorithms are presented in \cite{VV00,SV:cca,HVV:talg19}.
For the unweighted problem, there is a key lower bound of $n$ on $\opt$
(since any solution must have $\ge n$ edges, each of cost~one).
This fails to holds for MAP; indeed, the analogous lower bound on
$\opt$ is $\frac12 n$ for MAP. This rules out any direct extension
of these combinatorial methods (for the unweighted problem) to prove
approximation guarantees below~2 for MAP.

\subsection{Our results and techniques}

Our main contribution is a $\frac53$-approximation algorithm for
MAP, improving on the $\frac74$ approximation guarantee of
\cite{cdgkn:map}, see Theorems~\ref{thm:approxbydtwo},~\ref{thm:preproc}.

At a high level (hiding many important points), our algorithm is
based on a ``discharging scheme" where we compute a lower~bound on
$\opt$ (the optimal value) and fix a ``budget" of $\alpha$ times
this lower~bound (where $\alpha>1$ is a constant), ``scatter" this
budget over the graph $G$, use the budget to buy some edges to
obtain a ``base graph", then traverse the ``base graph" and buy
more edges to augment the ``base graph", so that (eventually) we
have a 2-ECSS whose cost is within the budget of $\alpha$ times our
lower~bound.
We mention that several of the results cited above are based on
discharging schemes, e.g.,
\cite{VV00,EFKN09,KN:talg16,HVV:talg19,cdgkn:map}.
In some more detail, but still at a high level, we follow the
method of \cite{cdgkn:map}.  We first pre-process the input instance
$G$, with the goal of removing all ``obstructions"
(e.g., cut~nodes), and we decompose $G$ into a list of ``well
structured" sub-instances $G_1, G_2, \dots$ that are pairwise
edge-disjoint.  Now, consider one of these sub-instances $G_i$ (it
has none of the ``obstructions").
We compute a subgraph $H_i$ whose cost is a lower bound on
$\opt(G_i)$.  Finally, we augment $H_i$ to make it 2-edge connected,
and use a credit-based analysis to prove an approximation guarantee.

Although our algorithm may appear to be similar to
the algorithm of \cite{cdgkn:map}, most of the details of the
algorithm and the analysis have been ``streamlined," and we have
``bypassed" the most difficult parts of the previous algorithm and
analysis.  Indeed, our presentation in this paper can be read
independently of \cite{cdgkn:map}.  (We have repeated a few definitions
and statements of results from \cite{cdgkn:map}.)

A 2-edge~cover is a subgraph that has at least two edges incident
to every node.  The minimum-cost 2-edge~cover is the key subgraph
used as a lower bound in our algorithm; we refer to it as \DTWO.
(\DTWO\ can be computed in polynomial time via extensions of Edmonds'
algorithm for computing a minimum-cost perfect matching.) Since
every 2-ECSS is a 2-edge~cover, we have $\cost(\DTWO)\leq\opt$.
So, by transforming \DTWO\ to a 2-ECSS of cost 
$\leq\frac53 \cost(\DTWO)$, we achieve our claimed approximation
guarantee.

Our pre-processing includes several new ideas, and moreover, it is
essential to handle new ``obstructions" that are not handled in
\cite{cdgkn:map}; indeed, \cite{cdgkn:map} has tight examples
such that $\opt/\cost(\DTWO) \ge \frac74 - \epsilon$ (for some $\epsilon>0$).
Although our algorithm handles several new ``obstructions", our
analysis and proofs for the pre-processing are simple.
One of our key tools (for our pre-processing analysis) is to prove
a stronger guarantee of $\max(\opt,\frac53\opt-2)$ rather than just
$\frac53\opt$. When we analyze our decomposition of an instance
into sub-instance(s), then this additive term of $-2$ is useful in
combining solutions back together at the end of the algorithm (when
we ``undo" the decomposition of $G$ into sub-instances $G_1,G_2,\dots$).

Our main algorithm (following \cite{cdgkn:map}) has two key subroutines
for transforming a \DTWO\ of a ``well structured" sub-instance
$G_i$ to a 2-ECSS of $G_i$ while ensuring that the total cost is
$\leq\frac53\cost(\DTWO)$.
\begin{enumerate}[(i)]
\item \textbf{Bridge covering step}:
The goal is to augment edges such that each connected component of
our ``current solution graph" $H_i$ is 2-edge-connected;
we start with $H_i:=\DTWO(G_i)$.
Our analysis is a based on a new and simple credit scheme that
bypasses some difficulties in the credit scheme of \cite{cdgkn:map}.
The most difficult part of the bridge covering subroutine of
\cite{cdgkn:map} handles a particular ``obstruction" that we call
a \usplit, see \cite[Lemma~24]{cdgkn:map} and see Section~\ref{s:prelims};
we ``eliminate" \usplit{}s during our pre-processing, thus, we
bypass the most difficult part of \cite{cdgkn:map}.

\item \textbf{Gluing step}:
Finally, this step merges the (already 2-edge connected) connected
components of $H_i$ to form a 2-ECSS of the sub-instance $G_i$.
A key part of this step handles so-called ``small 2ec-blocks";
these are cycles of cost~2 that occur as connected components of
$\DTWO(G_i)$
and stay unchanged through the bridge~covering step.
Observe that a ``small 2ec-block"
has only $\frac43$ credits (it has a ``budget" of $\frac{5}{3}(2)$,
and after paying for its two unit-edges, there is only $\frac43$
credits available).  Our gluing~step applies a careful swapping of
unit-edges for the ``small 2ec-blocks" while it merges the connected
components of $H_i$ into a 2-ECSS, and ensures that the net
augmentation cost does not exceed the available credit.

\end{enumerate}

\section{ \label{s:prelims}  Preliminaries}

This section has definitions and preliminary results.
Our notation and terms are consistent with \cite{Diestel},
and readers are referred to that text for further information.

Let $G=(V,E)$ be a (loop-free) multi-graph with edges of cost either zero or one
such that the edges of cost~zero form a matching.
We take $G$ to be the input graph, and
we use $n$ to denote $|V(G)|$.
Let $M$ denote the set of edges of cost~zero.
Throughout, the reader should keep in mind that $M$ is a matching;
this fact is used in many of our proofs without explicit reminders.
We call an edge of $M$ a \textit{zero-edge}
and we call an edge of $E-M$ a \textit{unit-edge}.

We denote the cost of an edge $e$ of $G$ by $\cost(e)$.
For a set of edges $F\subseteq E(G)$, $\cost(F):=\sum_{e\in F}\cost(e)$,
and for a subgraph $G'$ of $G$, $\cost(G'):=\sum_{e\in E(G')}\cost(e)$.

For ease of exposition, we often denote an instance $G,M$ by $G$;
then, we do not have explicit notation for the edge~costs of the instance,
but the edge~costs are given implicitly by $\cost:E(G)\rightarrow\{0,1\}$,
and $M$ is given implicitly by $\{e\in{E(G)}:\cost(e)=0\}$.

For a positive integer $k$, we use $[k]$ to denote the set $\{1,\dots,k\}$.

We use the standard notion of contraction of an edge, see \cite[p.25]{Schrijver}:
Given a multi-graph $H$ and an edge $e=vw$,
the contraction of $e$ results in the multi-graph $H/(vw)$ obtained from $H$
by deleting $e$ and its parallel copies and identifying the nodes $v$ and $w$.
(Thus every edge of $H$ except for $vw$ and its parallel copies
is present in $H/(vw)$; we disallow loops in $H/(vw)$.)

For a graph $H$ and a set of its nodes $S$,
$\Gamma_H(S):=\{w\in{V(H)-S}\,:\,v\in{S},vw\in{E(H)}\}$, thus,
$\Gamma_H(S)$ denotes the set of neighbours of $S$.

For a graph $H$ and a set of nodes $S\subseteq V(H)$,
$\delta_H(S)$ denotes the set of edges that have one end node in
$S$ and one end node in $V(H)-S$;
moreover,
$H[S]$ denotes the subgraph of $H$ induced by $S$, and
$H-S$  denotes the subgraph of $H$ induced by $V(H)-S$.
For a graph $H$ and a set of edges $F\subseteq E(H)$,
$H-F$ denotes the graph $(V(H),~E(H)-F)$.
We may use relaxed notation for singleton sets, e.g.,
we may use $\delta_H(v)$ instead of $\delta_H(\{v\})$, and
we may use $H-v$ instead of $H-\{v\}$, etc.

For any subgraph $K$ of a graph $H$ with $V(K)\subsetneq{V(H)}$,
an \textit{attachment} of $K$ is a node of $K$ that has a neighbour
in $V(H)-V(K)$.

We may not distinguish between a subgraph and its node~set;
for example, given a graph $H$ and a set $S$ of its nodes, we use
$E(S)$ to denote the edge~set of the subgraph of $H$ induced by $S$.

\subsection{2EC, 2NC, bridges and \DTWO}

A multi-graph $H$ is called $k$-edge connected if $|V(H)|\ge2$ and for
every $F\subseteq E(H)$ of size $<k$, $H-F$ is connected.
Thus, $H$ is 2-edge connected if it has $\ge2$ nodes and the deletion
of any one edge results in a connected graph.
A multi-graph $H$ is called $k$-node connected if $|V(H)|>k$ and for
every $S\subseteq V(H)$ of size $<k$, $H-S$ is connected.
We use the abbreviations \textit{2EC} for ``2-edge connected," and
\textit{2NC} for ``2-node connected."

We assume w.l.o.g.\ that the input $G$ is 2EC.
Moreover, for some (but not all) of our discussions, we assume that
there are $\leq 2$ copies of each edge (in the multi-graph under consideration);
this is justified since an edge-minimal 2-ECSS cannot have
three or more copies of any edge (see Proposition~\ref{propo:2ecdiscard} below).

For any instance $H$, let $\opt(H)$ denote the minimum cost of a
2-ECSS of $H$.  When there is no danger of ambiguity, we use $\opt$
rather than $\opt(H)$.

By a \textit{bridge} we mean
an edge of a connected (sub)graph whose removal results in two
connected~components, and by a \textit{cut~node} we mean a node of a
connected (sub)graph whose deletion results in
two or more connected~components.
We call a bridge of cost~zero a \textit{zero-bridge} and
we call a bridge of cost~one a \textit{unit-bridge}.

By a \textit{2ec-block} we mean a
maximal connected subgraph with two or more nodes that has no bridges.
(Observe that each 2ec-block of a graph $H$ corresponds to a
connected~component of order $\ge2$ of the graph obtained from $H$
by deleting all bridges.)
We call a 2ec-block \textit{pendant} if it is incident to exactly
one bridge.
We call a 2ec-block \textit{small} if it has $\le2$ unit-edges, and
we call it \textit{large} otherwise.

For a 2EC graph $G$ and a cut~node $v$ of $G$, a 2ec-$v$-block means
the subgraph of $G$ induced by $\{v\}\cup{V(C)}$ where $C$ is one
of the connected components of $G-v$.

The next result characterizes edges that are not essential for 2-edge~connectivity.

\begin{proposition} \label{propo:2ecdiscard}
Let $H$ be a 2EC graph and let $e=vw$ be an edge of $H$.
If $H-e$ has two edge-disjoint $v,w$~paths, then $H-e$ is 2EC.
\end{proposition}

The next lemma partially characterizes the cuts of size $\leq2$
in a graph obtained by ``uncontracting" a set of nodes of a 2EC graph.
It is our main tool for the analysis of our pre-processing steps.

\begin{lemma} \label{lem:uncontract}
Let $H$ be a 2EC graph and let $C\subsetneq V(H)$ be a set of
nodes such that the induced subgraph $H[C]$ is connected.
Suppose that $H^*$ is a 2-ECSS of $H/C$.
Let $H'$ be the spanning subgraph of $H$ with edge~set $E(C)\cup{E(H^*)}$.
Then $H'$ is a connected graph such that each of its bridges (if
any) is in $E(C)$.
\end{lemma}

\begin{proof}
In the graph $H'$, observe that for every node~set $S$ such that
$\emptyset\neq{S}\subseteq{V(H)-C}$,
we have $|\delta_{H'}(S)|\ge2$, because
$\delta_{H'}(S) = \delta_{H'/C}(S) = \delta_{H^*}(S)$ 
and $|\delta_{H^*}(S)|\geq2$ since $H^*$ is 2EC.
Similarly, for every node~set $S$ such that
$C\subseteq{S}\neq{V(H)}$,
we have $|\delta_{H'}(S)|\ge2$.
For any other set of nodes $S$ of $H'$, with $\emptyset\neq{S}\neq{V(H)}$,
we have $|\delta_{H'}(S)|\ge1$, because
both $S\cap{C}$ and $(V(H)-S)\cap{C}$ are nonempty, hence,
$\delta_{H'}(S) \supseteq \delta_{H[C]}(S\cap{C}) = \delta_{H[C]}(C-S)$
and $|\delta_{H[C]}(S\cap{C})|\geq1$ since $H[C]$ is connected.

In other words,
every cut $\delta(S)$ of $H'$, with $\emptyset\neq{S}\neq{V(H)}$,
has size $\ge2$ except the cuts that consist of a single edge of $H[C]$;
$H'$ is connected since none of these cuts is empty, and if $H'$
is not 2EC, then each of its bridges is an edge of $H(C)$.
\end{proof}

By a \textit{2-edge~cover} (of $G$) we mean
a set of edges $F$ of $G$ such that
each node $v$ is incident to at least two edges of $F$
(i.e., $F\subseteq E(G): |\delta_F(v)|\ge2, \forall v\in{V(G)}$).
By $\DTWO(G)$ we mean any minimum-cost 2-edge~cover of $G$
($G$ may have several minimum-cost 2-edge~covers, and $\DTWO(G)$
may refer to any one of them);
when there is no danger of ambiguity, we use \DTWO\ rather than $\DTWO(G)$.

By a \textit{bridgeless 2-edge~cover} (of $G$) we mean a 2-edge~cover
(of $G$) that has no bridges.

The next result follows from Theorem~34.15 in \cite[Chapter~34]{Schrijver}.

\begin{proposition} \label{thm:computeD2}
There is a polynomial-time algorithm for computing \DTWO.
\end{proposition}

The next result states the key lower~bound used by our approximation algorithm.

\begin{lemma} \label{lem:dtwolb}
Let $H$ be any 2EC graph. Then we have $opt(H) \geq \cost(\DTWO(H))$.
\end{lemma}

For any fixed positive integer $z$ (thus, $z=O(1)$) and any instance
of MAP, in time $O(1)$, we can determine whether the instance
has $\opt > {z}$, and if not, then we can find an optimal 2-ECSS
of the instance.

\begin{lemma} \label{lem:computeopt}
Let $H$ be an instance of MAP, and let $z$ be a fixed positive integer.
There is an $O(1)$-time algorithm to determine whether $\opt(H)
\geq z$.  Moreover, if $\opt(H) \leq z$, then a minimum-cost 2-ECSS
of $H$ can be found in $O(1)$ time.
\end{lemma}
\begin{proof}
Observe that $\opt(H)\geq |V(H)|/2$, because every 2-ECSS of $H$
has $\geq |V(H)|/2$ unit-edges; to see this, note that every 2-ECSS
of $H$ has $\geq |V(H)|$ edges and $H$ has $\leq |V(H)|/2$ zero-edges.

Our algorithm starts by checking whether $|V(H)|/2 \geq z$, and if
that holds, then clearly $\opt(H)\ge z$.
Otherwise, $|V(H)| < 2z$ (where $z=O(1)$), and our algorithm computes
$\opt(H)$; see the following discussion.

Suppose $|V(H)| < 2z$.
Note that the maximum size of an edge-minimal 2-ECSS of $H$ is $\le2|V(H)|-2$.
For each $k=1,\dots,2|V(H)|-2,$
the algorithm examines each set of unit-edges $F\subseteq E(H)$ of size $k$,
and checks whether $F\cup{M}$ is a 2-ECSS of $H$;
recall that $M$ denotes the set of zero-edges of $H$.
Clearly, $\opt(H)$ is given by the smallest $k=|F|$ such that
$F\cup{M}$ is a 2-ECSS of $H$, and the corresponding $F\cup{M}$ is
an optimal 2-ECSS of $H$.

The algorithm runs in time $O(2^{|E(H)|} |E(H)|) = O(1)$
since $|E(H)| \le |V(H)|^2 < 4z^2 = O(1)$.
\end{proof}

\subsection{Obstructions for the approximation guarantee}

There are several obstructions (e.g., cut~nodes)
that prevent our algorithm (and analysis)
from achieving our target approximation factor of $\frac53$.
We eliminate all such obstructions in a pre-processing step
that takes the given instance $G$ of MAP (the input)
and replaces it by a list of sub-instances $G_1,G_2,\dots,$ such that
(a)~none of the obstructions occurs in a sub-instance $G_i$,
(b)~the edge-sets of the sub-instances are pairwise-disjoint,
and
(c)~given a 2-ECSS of each sub-instance $G_i$ of approximately optimal cost,
we can construct a 2-ECSS of $G$ of cost $\le\frac53\opt(G)$.
(Precise statements are given later.)
The obstructions for our algorithm are:
\\
 \begin{enumerate}[(i)]
 \item  cut nodes,
 \item  parallel edges,
 \item  \zsplit,
 \item  \usplit,
 \item  \csplit,
 \item  \redcycle,
 \item  \redgadget.
 \end{enumerate}

Below, we formally define each of these obstructions.
Four of these obstructions were introduced in \cite{cdgkn:map}, and
readers interested in a deeper understanding may refer to that
paper, in particular, see the remark after \cite[Theorem~6]{cdgkn:map}
and see \cite[Figure~2]{cdgkn:map} for instances $G$ of MAP that
contain cut~nodes, parallel~edges, \zsplit{}s, or \redcycle{}s such
that $\opt(G)/\cost(\DTWO(G))\approx2$; informally speaking, an
approximation algorithm based on the lower~bound $\cost(\DTWO(G))$
on $\opt(G)$ fails to beat the approximation threshold of $2$ in
the presence of any of these four obstructions.
In an appendix (see Section~\ref{s:appendix}), we present instances
$G$ of MAP that contain either the \csplit\ obstruction or the
\redgadget\ obstruction (and none of the other six obstructions)
such that $\opt(G)/\cost(\DTWO(G))\approx\frac74$;
informally speaking, an approximation algorithm based on the
lower~bound $\cost(\DTWO(G))$ on $\opt(G)$ fails to beat the
approximation threshold of $7/4$ in the presence of any of these
two obstructions.
The remaining obstruction, \usplit, is relevant for our bridge~covering
step and its analysis; see the proof of
Proposition~\ref{propo:find-pseudo-ear}.  As mentioned before, by
``eliminating" \usplit{}s during our pre-processing, we bypass a
difficult part of \cite{cdgkn:map}, see \cite[Lemma~24]{cdgkn:map}.

\begin{defn}
By a \zsplit\ (also called a {bad-pair}),
we mean a zero-edge $e$ and its end~nodes, $u,v$,
such that $G-\{u,v\}$ has $\ge2$ connected components.
\end{defn}

\begin{defn}
By a \usplit,
we mean a unit-edge $e$ and its end~nodes, $u,v$,
such that $G-\{u,v\}$ has $\ge2$ connected components;
moreover, in the graph $G/\{u,v\}$,
there exist two distinct 2ec-$\hv$-blocks $B_1,B_2$ incident to the
contracted node $\hv$ such that
$\opt(B_i)\ge3$ and
$B_i$ has a zero-edge incident to the contracted node, $\forall{i}\in[2]$.
\end{defn}

\begin{defn} \label{d:csplit}
By an \csplit,
we mean an induced 2NC subgraph $C$ of $G$ with $|V(C)|\in\{3,4\}$
that has a spanning cycle of cost~two such that
$G-V(C)$ has $\ge2$ connected components, and
the cut $\delta(V(C))$ has no zero-edges;
moreover, in the graph $G/C$,
there exist two distinct 2ec-$\hv$-blocks $B_1,B_2$ incident to the
contracted node $\hv$ that have $\opt(B_1)\ge3$ and $\opt(B_2)\ge3$.
\end{defn}

\authremark{~The absence of \csplit{}s in instances of \wsMAP{} is
used only in Section~\ref{s:gluing}, see the proofs of
Lemmas~\ref{lem:smallblocks:a},~\ref{lem:Daux}.
Informally speaking, the presence of an \csplit\ $C$ in an instance
$G$ of MAP implies that there is a corresponding cycle $\hat{C}$
of cost~two with $|V(\hat{C})|=3$ or $|V(\hat{C})|=4$ such that
$G-V(\hat{C})$ is disconnected; our algorithm \& analysis for the
gluing~step could fail in the presence of an \csplit{},
see the appendix, Section~\ref{s:appendix}.
Moreover, there exist instances $G$ of MAP that contain \csplit{}s
and contain none of the other six obstructions such that
$\opt(G)/\cost(\DTWO(G))\approx\frac74$,
see the appendix, Section~\ref{s:appendix}, and
see \cite[Section~7.1]{cdgkn:map}.}

\begin{defn}
By an \redcycle\ (also called a {redundant 4-cycle}),
we mean an induced subgraph $C$ of $G$ with four nodes such that
$V(C)\not=V(G)$, $C$ contains a 4-cycle of cost~two, and $C$ contains
a pair of nonadjacent nodes that each have degree~two in $G$.
\end{defn}

\begin{defn} \label{d:redgadget}
By an \redgadget,
we mean an induced subgraph $C$ of $G$ with eight nodes such that
$V(C)\not=V(G)$,
$C$ contains two disjoint 4-cycles $C_1$, $C_2$ with $\cost(C_i)=2,\forall{i}\in[2]$,
$C$ has exactly two attachments $a_1,a_2$ where $a_i\in{C_i},\forall{i}\in[2]$, and
both end~nodes of the (unique) unit-edge of $C_i-a_i$ are adjacent
to $C_{3-i}$, $\forall{i}\in[2]$.
\end{defn}

\authremark{~The absence of \redgadget{}s in instances of \wsMAP{}
is used only in Section~\ref{s:gluing}, see the proof of
Lemma~\ref{lem:Daux}.
Moreover, there exist instances $G$ of MAP that contain \redgadget{}s
and contain none of the other six obstructions such that
$\opt(G)/\cost(\DTWO(G))\approx\frac74$,
see the appendix, Section~\ref{s:appendix}.}

See Figures~\ref{f:0split-Ex}, \ref{f:1split-Ex}, \ref{f:4split-Ex},
and~\ref{f:R8Ex} for illustrations of \zsplit{}s, \usplit{}s,
\csplit{}s, and \redgadget{}s, respectively.

\input{Figures/figures-4obstructions.tex}

\subsection{Polynomial-time computations}

There are well-known polynomial~time algorithms for implementing
all of the basic computations in this paper, see \cite{Schrijver}.
We state this explicitly in all relevant results
(e.g., Theorem~\ref{thm:approxbydtwo}),
but we do not elaborate on this elsewhere.

\section{ \label{s:algo} Outline of the algorithm}

This section has an outline of our algorithm.
We start by defining an instance of \wsMAP{}.

\begin{defn}
An instance of \wsMAP{} is an instance of MAP with $\ge\Nmin$ nodes that contains
\begin{multicols}{2}
 \begin{itemize}
 \item[-] no cut nodes,
 \item[-] no parallel edges,
 \item[-] no \zsplit,
 \item[-] no \usplit,
 \item[-] no \csplit,
 \item[-] no \redcycle,
and
 \item[-] no \redgadget.
 \item[\vspace{\fill}]
 \end{itemize}
\end{multicols}
\end{defn}

In this section and
Section~\ref{s:pre-proc}, we explain
how to ``decompose'' any instance of MAP $G$ with $|V(G)|\ge\Nmin$
into a collection of instances $G_1,\dots,G_k$ of MAP such that
(a)~either $|V(G_i)| < \Nmin$ or $G_i$ is an instance of \wsMAP{}, $\forall i \in [k]$,
(b)~the edge~sets $E(G_1),\dots,E(G_k)$ are pairwise disjoint
(thus $E(G_1),\dots,E(G_k)$ forms a subpartition of $E(G)$),
and
(c)~a 2-ECSS $H$ of $G$ can be obtained
by computing 2-ECSSes $H_1,\dots,H_k$ of $G_1,\dots,G_k$.
Moreover,
the approximation guarantee is preserved,
meaning that $\cost(H) \leq \frac53 \opt(G)-2$
provided $\cost(H_i) \leq \max(\opt(G_i),\,\frac53\opt(G_i)-2), \forall{i}\in[k]$.

\medskip
\noi
\fbox{ \begin{minipage}{0.9\textwidth}

\textbf{Algorithm (outline)}:
{
\setlength{\itemsep}{0pt}
\begin{itemize}
\item[(0)] apply the pre-processing steps
	(see below and see Section~\ref{s:pre-proc})
	to obtain a collection of instances {$G_1,\dots,G_k$}
		such that either $|V(G_i)| < \Nmin$ or
		$G_i$ is an instance of \wsMAP{}, $\forall i \in [k]$;
\item[] \textbf{for} each $G_i$ ($i=1,\dots,k$), 
\item[] \textbf{if} $|V(G_i)| < \Nmin$
\item[(1)] exhaustively compute an optimum 2-ECSS $H_i$ of $G_i$
	via Lemma~\ref{lem:computeopt};
\item[] \textbf{else}
\item[(2.1)] compute \DTWO($G_i$) in polynomial time
	(w.l.o.g.\ assume \DTWO($G_i$) contains all zero-edges of~$G_i$);
\item[(2.2)] then apply ``bridge~covering" from Section~\ref{s:bridge-cover}
	to \DTWO($G_i$) to obtain a bridgeless 2-edge~cover $\wt{H}_i$ of $G_i$;
\item[(2.3)] then apply the ``gluing~step" from Section~\ref{s:gluing}
	to $\wt{H}_i$ to obtain a 2-ECSS $H_i$ of $G_i$;
\item[] \textbf{endif};
\item[] \textbf{endfor};
\item[(3)] finally, output a 2-ECSS $H$ of $G$ from the union of
	$H_1,\dots,H_k$ by undoing the transformations applied in step~(0).
\end{itemize}
}
\end{minipage}
}
\medskip

The pre-processing of step~(0) consists of several reductions;
most of these reductions are straightforward, but we have to prove that
the approximation guarantee is preserved when we ``undo" each of these reductions.
These proofs are given in Section~\ref{s:pre-proc}.

\medskip
\noi
\fbox{ \begin{minipage}{0.9\textwidth}

\textbf{Pre-processing -- Step~(0) of Algorithm}:
\\[0.25ex]

\textbf{While} the current list of sub-instances $G_1,G_2,\dots$ has a
sub-instance $G_i$ that has $\ge\Nmin$ nodes and is not an instance
of \wsMAP{} (assume that $G_i$ is 2EC):

{
\setlength{\itemsep}{0pt}
\begin{itemize}

\item[] {\textbf{if} $G_i$ is not 2NC:}

\item[(i)] (handle a cut-node) \\
let $v$ be a cut~node of $G_i$, and let 
$B_1,\dots,B_k$ be the 2ec-$v$-blocks of $G_i$;
replace $G_i$ by $B_1,\dots,B_k$ in the current list;

\item[] {\textbf{else} apply exactly one of the following steps to $G_i$:}

\item[(ii)] (handle a pair of parallel edges) \\
let $\{e,f\}$ be a pair of parallel edges of $G_i$
(one of the edges in $\{e,f\}$ is a unit-edge);
discard a unit-edge of $\{e,f\}$ from $G_i$;

\item[(iii)] (handle an ``S obstruction")
\begin{enumerate}[(a)]
\item (handle a \usplit) 
\item (handle a \zsplit)
\item (handle an \csplit) 
\end{enumerate}

let $C$ denote a subgraph of $G_i$ that is, respectively,
(a)~a \usplit, (b)~a \zsplit, or (c)~an \csplit;

contract $C$ to obtain $G_i/C$ and let $\hv$ denote the contracted node;
let $B_1,\dots,B_k$ be the 2ec-$\hv$-blocks of $G_i/C$;
replace $G_i$ by $B_1,\dots,B_k$ in the current list;

\item[(iv)] (handle an ``R obstruction")
\begin{enumerate}[(a)]
\item (handle an \redcycle)
\item (handle an \redgadget)
\end{enumerate}

let $C$ denote a subgraph of $G_i$ that is, respectively,
(a)~an \redcycle, or~(b) an \redgadget;

contract $C$ to obtain $G_i/C$, and
replace $G_i$ by $G_i/C$ in the current list;

\end{itemize}
}
\end{minipage}
}
\medskip

Our $\frac53$ approximation algorithm for MAP follows from 
the following theorem;
our proof is given in Section~\ref{s:gluing} (see
page~\pageref{prf:approxbydtwo}).

\begin{theorem}
\label{thm:approxbydtwo}
Given an instance of \wsMAP{} $G'$,
  there is a polynomial-time algorithm that
  obtains a 2-ECSS $H'$
  such that $\cost(H') \leq \max( \opt(G'),\; \frac{5}{3} \opt(G') - 2 )$.
\end{theorem}

We use a credit~scheme to prove this theorem; the details are
presented in Sections~\ref{s:bridge-cover} and~\ref{s:gluing}.
The algorithm starts with $\DTWO(G')$ as the current graph, and
assigns $\frac53$ tokens to each unit-edge of $\DTWO(G')$; each
such edge keeps one unit to pay for itself and the other $\frac23$
is taken to be credit of the edge; thus, the algorithm has $\frac23
\cost(\DTWO(G'))$ credits at the start; the algorithm uses the credits
to pay for the augmenting edges ``bought" in steps~(2.2) or~(2.3)
(see the outline); also, the algorithm may ``sell" unit-edges of
the current~graph (i.e., such an edge is permanently discarded and
is not contained in the 2-ECSS output by the algorithm).

The factor $\frac53$ in our approximation guarantee is tight
in the sense that there exists an instance $G$ of \wsMAP{}
such that $\opt(G)/\cost(\DTWO(G)) \ge \frac53-\epsilon$,
for any small positive number $\epsilon$.
The instance $G$ consists of a root 2ec-block $B_0$, say a 6-cycle
of cost~6, $v_1,\dots,v_6,v_1$, and $\ell\gg1$ copies of the following
gadget that are attached to $B_0$. The gadget consists of a 6-cycle
$C=u_1,\dots,u_6,u_1$ of cost~3 that has alternating zero-edges and
unit-edges; moreover, there are three unit-edges between $C$ and
$B_0$: $v_1u_1$, $v_3u_3$, $v_5u_5$. Observe that a (feasible)
2-edge~cover of this instance consists of $B_0$ and the 6-cycle $C$
of each copy of the gadget, and it has cost $6+3\ell$.  Observe
that for any 2-ECSS and for each copy of the gadget, the six edges
of $C$ as well as (at least) two of the edges between $C$ and $B_0$
are contained in the 2-ECSS. Thus, $\opt(G)\geq6+5\ell$, whereas
$\cost(\DTWO(G))\leq6+3\ell$.

\section{ \label{s:pre-proc} Pre-processing}
{
This section presents the proofs and analysis for the pre-processing
step of our algorithm.

We use $\alpha\ge\frac53$ to denote a positive real number
that is used in the analysis of our approximation guarantee;
we take $\alpha$ to be~$\frac53$ for our main result.
Informally speaking, most of the results in this section
prove an approximation guarantee of the form $(\alpha\;\opt-2)$.
The additive term of $-2$ is critical,
because when we undo the transformations applied in step~(0)
(see the outline of the algorithm in Section~\ref{s:algo}),
then we incur an additional cost of $+1$ or $+2$
(for example, when we undo the transformation for an \csplit,
then we incur the additional cost of $2$ for a spanning cycle of that \csplit);
in spite of this additional cost, we derive an 
approximation guarantee of the form $(\alpha\;\opt-2)$
by using the $-2$ term to compensate for the additional cost.
But note that $\alpha\;\opt-2$ is an invalid approximation guarantee
whenever $\opt\leq2$ (since $\alpha\;\opt-2 < \opt$ for $\opt\leq2$).
In fact, our approximation guarantees have the form $\max(\opt,\;\alpha\;\opt-2)$.

\begin{lemma} \label{lem:findobs}
Every occurrence of each of the seven types of obstructions
(i.e., cut~nodes,
parallel edges,
\zsplit, \usplit, \csplit, \redcycle, \redgadget)
can be computed in polynomial time.
\end{lemma}
\begin{proof}
Each type of obstruction is a subgraph on $O(1)$ nodes.
A simple method is to exhaustively check each subset of nodes $S$
of the appropriate cardinality and decide whether or not
the relevant properties hold for the subgraph induced by $S$.

There are better algorithms for some types of obstructions, e.g.,
there is a linear-time algorithm for computing all the cut~nodes.
\end{proof}

The following lemmas address the pre-processing and post-processing
(that is, steps~(0) and~(3) of the outline)
of each of the seven types of obstructions.
\big(Some of the proofs use the following observation: Suppose that
a 2EC graph $M$ has a cut node $v$ and has 2ec-$v$-blocks $M_1,\dots,M_k$.
Any 2-ECSS $M'$ of $M$ induces a 2-ECSS on each of
$V(M_1),\dots,V(M_k)$, hence, $\opt(M)=\sum_{i=1}^k\opt(M_i)$.\big)

\begin{lemma} \label{lem:pp:cutnodes}
Let $v$ be a cut~node of $G$, and let 
$B_1,\dots,B_k$ be the 2ec-$v$-blocks of $G$.
Let $B'_1,\dots,B'_k$ be 2-ECSSs of $B_1,\dots,B_k$
such that $\cost(B'_i) \leq \max( \opt(B_i),\;
	\alpha\;\opt(B_i)-2),\;\forall{i}\in[k]$.
Then $B'_1\cup\dots\cup{B'_k}$ is a 2-ECSS of $G$ of cost
$\leq \max( \opt(G),\; \alpha\;\opt(G)-2)$.  
\end{lemma}
\begin{proof}
By {Lemma}~\ref{lem:uncontract}, $B'_1\cup\dots\cup{B'_k}$ is a 2-ECSS of $G$.

We have $\opt(G)=\sum_{i=1}^k\opt(B_i)$.
If $\cost(B'_i) \leq \opt(B_i),\;\forall{i}\in[k]$, then
$\cost(B'_1\cup\dots\cup{B'_k})\leq \sum_{i=1}^k\opt(B_i)=
	\opt(G)\leq\max(\opt(G),\; \alpha\;\opt(G)-2)$.
Otherwise, there is a $j\in[k]$ with $\opt(B_j)<\alpha\;\opt(B_j)-2$,
then $\cost(B'_i)\leq\alpha\;\opt(B_i),\;\forall{i}\in[k],i\neq{j},$
and $\cost(B'_j)\leq\alpha\;\opt(B_j)-2$, hence,
$\cost(B'_1\cup\dots\cup{B'_k})\leq\alpha\;\opt(G)-2$.
\end{proof}

\begin{lemma} \label{lem:pp:paredges}
Let $e,f$ be a pair of parallel edges of a 2NC graph $G$, and let
$f$ be a unit-edge.
Let $B'$ be a 2-ECSS of $G-f$ of cost $\leq \max(\opt(G-f),\; \alpha\;\opt(G-f)-2)$.
Then $B'$ is a 2-ECSS of $G$  of cost $\leq \max(\opt(G),\; \alpha\;\opt(G)-2)$.
\end{lemma}
\begin{proof}
The result holds because a 2NC graph has an optimal 2-ECSS that
contains no parallel edges; this can be proved using the arguments
used to prove \cite[Fact~8]{cdgkn:map}.
Hence, we have $\opt(G)=\opt(G-f)$.
\end{proof}

\begin{lemma} \label{lem:pp:zsplit}
Let $e=uv$ be a \zsplit\ of a 2NC graph $G$, and let 
$B_1,\dots,B_k$ be the 2ec-$\hv$-blocks of $G/e$, where
$\hv$ denotes the contracted node of $G/e$.
Let $B'_1,\dots,B'_k$ be 2-ECSSs of $B_1,\dots,B_k$
such that $\cost(B'_i) \leq \max( \opt(B_i),\;
	\alpha\;\opt(B_i)-2),\;\forall{i}\in[k]$.
Then there exist an index $i\in[k]$
and $F'_i\subseteq E(B_i)$ of cost $\leq \cost(B'_i)+1$ such that
$\{e\}\cup{E(B'_1)}\cup\dots\cup{E(B'_{i-1})}\cup{F'_i}\cup{E(B'_{i+1})}\cup\dots\cup{E(B'_k)}$
is (the edge~set of) a 2-ECSS of $G$ of cost $\leq \max( \opt(G),\; \alpha\;\opt(G)-2)$;
moreover, $F'_i$ can be computed from $E(B'_i)$ in $O(|V(G)|)$ time.
\end{lemma}
\begin{proof}
For all $i\in[k]$, observe that $\opt(B_i)\ge2$
because all edges incident to $\hv$ in $B_i$ are unit-edges.
First, suppose that there is an index $i\in[k]$, say $i=1$,
with $\opt(B_i)=2$.
Then $B_1$ has $\leq3$ nodes
(since $B_1$ has an optimal 2-ECSS that has cost zero on $E(B_1-\hv)$).
Let ${B_1^{\oplus}}$ be the subgraph of $G$ induced by $\{u,v\}\cup (V(B_1)-\hv)$.
Then ${B_1^{\oplus}}$ is a 2NC graph (see \cite[Fact~14]{cdgkn:map}) and it has
$3$ nodes or $4$ nodes; moreover, ${B_1^{\oplus}}$ has a spanning cycle $\hat{C_1}$ of cost~two.
We replace $E(B'_1)$ by $E(\hat{C_1})-\{uv\}$.
Then the spanning subgraph $H'$ with edge~set
${E(\hat{C_1})}\cup{E(B'_2)}\cup\dots\cup{E(B'_k)}$ is a 2-ECSS of $G$ 
(by {Lemma}~\ref{lem:uncontract} and the fact that $\hat{C_1}$ contains $e$),
and $\cost(H') \leq \max( \opt(G),\; \alpha\;\opt(G)-2)$.

Now, suppose that $\opt(B_i)\ge3$ for all $i\in[k]$.
By {Lemma}~\ref{lem:uncontract},
the spanning subgraph $H'$ with edge~set
$\{e\}\cup{E(B'_1)}\cup\dots\cup{E(B'_k)}$
has at most one bridge, namely, $e$.
If $e$ is a bridge of $H'$, then no edge of $E(B'_1)$ is incident
to one of the end~nodes of $e$, say $u$.  Pick $f$ to be any edge of
$G$ between $V(B'_1)-\hv$ and $u$.
($G$ has such an edge, otherwise, $v$ would be a cut~node of $G$.)
Clearly, adding $f$ to $H'$ results in a 2-ECSS of $G$.

We have $\opt(G)\geq\sum_{i=1}^k\opt(B_i)$.
Then, $\cost(H'\cup\{f\}) = 1+\sum_{i=1}^k\cost(B'_i) \leq
\max(\opt(G),\; \alpha\;\opt(G)-2)$, because either 
$\cost(B'_i)\leq\alpha\;\opt(B_i)-2$ holds for two indices in $[k]$ or
there is an index $i\in[k]$ with $\cost(B'_i)=\opt(B_i)\geq3$
and so $\cost(B'_i)+1\leq\alpha\;\opt(B_i)$.
\end{proof}

\begin{lemma} \label{lem:pp:usplit}
Let $e=uv$ be a \usplit\ of a 2NC graph $G$, and let 
$B_1,\dots,B_k$ be the 2ec-$\hv$-blocks of $G/e$, where
$\hv$ denotes the contracted node of $G/e$.
Let $B'_1,\dots,B'_k$ be 2-ECSSs of $B_1,\dots,B_k$
such that $\cost(B'_i) \leq \max( \opt(B_i),\;
	\alpha\;\opt(B_i)-2),\;\forall{i}\in[k]$.
Then there exists an edge $f$ of $G$ such that
$\{e,f\}\cup{E(B'_1)}\cup\dots\cup{E(B'_k)}$ is (the edge~set of)
a 2-ECSS of $G$ of cost $\leq \max( \opt(G),\; \alpha\;\opt(G)-2)$.
\end{lemma}
\begin{proof}
By Lemma~\ref{lem:uncontract},
the spanning subgraph $H'$ with edge~set
$\{e\}\cup{E(B'_1)}\cup\dots\cup{E(B'_k)}$
has at most one bridge, namely, $e$.
If $e$ is a bridge of $H'$, then no edge of $E(B'_1)$ is incident
to one of the end~nodes of $e$, say $u$.  Pick $f$ to be any edge of
$G$ between $V(B'_1)-\hv$ and $u$.  
($G$ has such an edge, otherwise, $v$ would be a cut~node of $G$.)
Clearly, adding $f$ to $H'$
results in a 2-ECSS of $G$.

We have $\opt(G)\geq\sum_{i=1}^k\opt(B_i)$.
Then, $\cost(E(H')\cup\{f\}) = \cost(\{e,f\})+\sum_{i=1}^k\cost(B'_i)
	= 2+\sum_{i=1}^k\cost(B'_i)$ $\leq$ $\max(\opt(G),\;$ $\alpha\;\opt(G)-2)$,
because $3\leq\opt(B_i)\leq\cost(B'_i)\leq\alpha\;\opt(B_i)-2$ holds
for two indices in $[k]$, by definition of a \usplit.
\end{proof}

\begin{lemma} \label{lem:pp:csplit}
Let $C$ be an \csplit\ of a 2NC graph $G$, and let 
$B_1,\dots,B_k$ be the 2ec-$\hv$-blocks of $G/C$, where
$\hv$ denotes the contracted node of $G/C$.
Let $B'_1,\dots,B'_k$ be 2-ECSSs of $B_1,\dots,B_k$
such that $\cost(B'_i) \leq \max( \opt(B_i),\;
	\alpha\;\opt(B_i)-2),\;\forall{i}\in[k]$.
Let $\hat{C}$ be a spanning cycle of $C$ of cost~two.
Then
$E(\hat{C})\cup{E(B'_1)}\cup\dots\cup{E(B'_k)}$ is (the edge~set of)
a 2-ECSS of $G$ of cost $\leq \max( \opt(G),\; \alpha\;\opt(G)-2)$.
\end{lemma}

\begin{proof}
Note that $\hat{C}$ is 2EC, so 
by Lemma~\ref{lem:uncontract},
the spanning subgraph $H'$ with edge~set
$E(\hat{C})\cup{E(B'_1)}\cup\dots\cup{E(B'_k)}$
is a 2-ECSS of $G$. 

We have $\opt(G)\geq\sum_{i=1}^k\opt(B_i)$.
Then, $\cost(H') = \cost(\hat{C})+\sum_{i=1}^k\cost(B'_i) = 2+\sum_{i=1}^k\cost(B'_i) \leq
\max(\opt(G),\; \alpha\;\opt(G)-2)$, because
$3\leq\opt(B_i)\leq\cost(B'_i)\leq\alpha\;\opt(B_i)-2$ holds for two indices in $[k]$,
by definition of an \csplit. 
\end{proof}

\begin{lemma} \label{lem:pp:redcycle}
Let $C$ be an \redcycle\ of a 2NC graph $G$.
Let $B'_1$ be a 2-ECSS of $G/C$
such that $\cost(B'_1) \leq \max( \opt(G/C),\;
	\alpha\;\opt(G/C)-2)$.
Then
$E(C)\cup{E(B'_1)}$ is (the edge~set of)
a 2-ECSS of $G$ of cost $\leq \max( \opt(G),\; \alpha\;\opt(G)-2)$.
\end{lemma}
\begin{proof}
Note that $C$ is 2EC, so 
by Lemma~\ref{lem:uncontract},
the spanning subgraph $H'$ with edge~set
$E(C)\cup{E(B'_1)}$
is a 2-ECSS of $G$. 

Recall that an \redcycle\ contains two nodes of degree exactly 2. 
In particular, any 2-ECSS of $G$ will contain all edges of $E(C)$,
so $\opt(G)\geq2+\opt(G/C)$. 
Then, $\cost(H') = 2+\cost(B'_1) \leq
\max(\opt(G),\; \alpha\;\opt(G)-2)$. 
\end{proof}

\begin{lemma} \label{lem:pp:redgadget}
Let $C$ be an \redgadget\ of a 2NC graph $G$ where $|V(G)|\geq12$.
Let $B'_1$ be 2-ECSS of $G/C$
such that $\cost(B'_1) \leq \max( \opt(G/C),\;
	\alpha\;\opt(G/C)-2)$.
Then there exists $F\subseteq E(C)$ of cost $\leq5$ such that
$F\cup{E(B'_1)}$ is (the edge~set of)
a 2-ECSS of $G$ of cost $\leq \max( \opt(G),\; \alpha\;\opt(G)-2)$.
\end{lemma}
\begin{proof}
Let $F$ be the edge~set of a 2-ECSS of $C$ of minimum cost.
Then $\cost(F)\leq5$.
(To see this, consider the two disjoint $4$-cycles $C_1,C_2$ of $C$
and let $e=uv\in E(C_1)$ be a unit-edge such that $u$ and $v$ are
incident to edges $f_1,f_2$, respectively, such that both $f_1$ and
$f_2$ have an end~node in $C_2$; let
$F=E(C_1)\cup{E(C_2)}\cup\{f_1,f_2\}-\{e\}$.)
By Lemma~\ref{lem:uncontract}, the spanning subgraph $H'$ with edge
set $F \cup E(B'_1)$ is a 2-ECSS of $G$.

Observe that $\opt(G) \geq 3+\opt(G/C)$, because any 2-ECSS of $G$
has $\ge 7$ edges of $C$ (since $C$ has 8~nodes and exactly two attachments),
and $\ge3$ of these edges have unit cost.
Moreover, since $G/C$ has $\geq5$ nodes,
$3\leq \opt(G/C)\leq \cost(B'_1) \leq \alpha\;\opt(G/C)-2$.
Hence, $\cost(H') = 5 + \cost(B'_1) \le \alpha\;\opt(G)-2$. 
\end{proof}

\begin{theorem} \label{thm:preproc}
Suppose that there is an approximation algorithm that
given an instance $H$ of \wsMAP{}, finds a 2-ECSS of cost
$\leq \max( \opt(H),\; \alpha\;\opt(H)-2)$.
Then, given an instance $G$ of MAP, there is a polynomial-time
algorithm to find a 2-ECSS of cost $\leq \max( \opt(G),\;
\alpha\;\opt(G)-2)$.
\end{theorem}

\begin{proof}
Let $n$ and $m$ denote $|V(G)|$ and $|E(G)|$.
First, observe that there are at most $O(n+m)$ iterations
of the while-loop of the pre-processing algorithm (given in the box).
To see this, consider the ``potential function" $\phi$ given by the
sum over all graphs $G_i$ in the current list of
$|E(G_i)|+\#\text{cutnodes}(G_i)$ (i.e., sum of the number of edges
of $G_i$ and the number of cut~nodes of $G_i$).  Initially, $\phi\leq
m+n$; $\phi$ decreases (by one or more) in every iteration because
each of the ``operations" (labelled by (i),~(ii),~(iii) (a),(b),(c),
(iv) (a),(b)) causes $\phi$ to decrease; $\phi$ is $\ge0$
always.  Hence, the number of iterations is $\leq m+n$.
Clearly, each iteration can be implemented in polynomial time.

The upper bound on the cost of the 2-ECSS solution
follows from the previous results in this section, i.e.,
Lemmas~\ref{lem:findobs}--\ref{lem:pp:redgadget}.
\end{proof}
}

\section{ \label{s:bridge-cover}  Bridge~covering}
{
The results in this section are based on the prior results and
methods of \cite{cdgkn:map,dippel:thesis}, but the goal in these
previous papers is to obtain an approximation guarantee of $\frac74$
for MAP, whereas our goal is an approximation guarantee of $\frac53$.
Our credit invariant is presented in Section~\ref{s:credits} below,
and it  is based on the credit invariant in \cite{dippel:thesis}.

In this section and in Section~\ref{s:gluing}, we assume that
the input is an instance of \wsMAP{}.
For notational convenience, we denote the input by $G$.
Recall that $G$ is a simple, 2NC graph on $\ge \Nmin$ nodes,
and $G$ has no \zsplit, no \usplit, no \csplit, no \redcycle, and no \redgadget.
Recall that a 2ec-block is called small if it has $\le2$ unit-edges,
and is called large otherwise.
Since $G$ is 2NC and simple, a small 2ec-block is either
a 3-cycle with one zero-edge and two unit-edges,
or a 4-cycle with alternating zero-edges and unit-edges.

Each unit-edge $e$ of \DTWO\ starts with $\frac53$ tokens, and from
this, one unit is kept aside (to pay for $e$), and the other $\frac23$
is defined to be the credit of $e$.
Our overall goal is to find a 2-ECSS $H'$ of $G$ of cost $\leq
\frac53 \cost(\DTWO)$, and we keep $\frac23 \cost(\DTWO)$ from our
budget in the form of credit while using the rest of our budget for
``buying" the unit-edges of \DTWO. We use the credit for ``buying"
unit-edges that are added to our current graph during the bridge~covering
step or the gluing~step.
(In the gluing~step, we may ``sell" unit-edges of our current graph,
that is, we may permanently discard some unit-edges of our current graph;
thus, our overall budgeting scheme does not rely solely on credits.)

We use $H$ to denote the current graph of the bridge~covering~step;
initially, $H=\DTWO$.

The outcome of the bridge~covering step is stated in the following result.

\begin{proposition} \label{propo:bridgecover}
At the termination of the bridge~covering~step, $H$ is a bridgeless
2-edge~cover;
moreover, every small 2ec-block of $H$ has $\ge\frac43$ credits and
every large 2ec-block of $H$ has $\ge2$ credits.
The bridge~covering~step can be implemented in polynomial time.
\end{proposition}

{
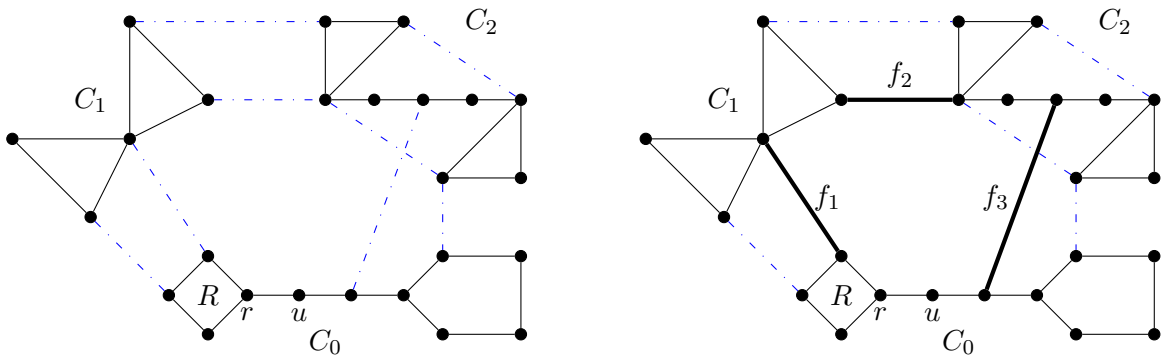
\begin{figure}[htb]
	\centering
	\begin{subfigure}{0.49\textwidth}
	\centering
    \begin{tikzpicture}[scale=0.52]
    	\begin{scope}[every node/.style=
    				 {circle, fill=black, draw, inner sep=0pt, 
    				  minimum size = 0.15cm}]

		\node[] (u1) at (-7, 2) {}; 
    		\node[] (u2) at (-5, 0) {}; 
		\node[] (u3) at (-4, 2) {}; 
    		\node[] (u4) at (-4, 5) {}; 
    		\node[] (u5) at (-2, 3) {}; 
            
    		\node[] (v1) at (-3, -2) {}; 
    		\node[] (v2) at (-2, -1) {}; 
    		\node[label={[label distance=1]270:$r$}]  (v3) at (-1, -2) {}; 
    		\node[] (v4) at (-2, -3) {}; 
            
    		\node[label={[label distance=1]270:$u$}] (p1) at (0.33, -2) {}; 
    		\node[] (p2) at (1.66, -2){}; 
            
    		\node[] (v5) at (3, -2) {}; 
		\node[] (v6) at (4, -1) {}; 
		\node[] (v7) at (6, -1) {}; 
		\node[] (v8) at (6, -3) {}; 
		\node[] (v9) at (4, -3) {}; 
			
		\node[] (w1) at (1, 3) {}; 
		\node[] (w2) at (1, 5) {}; 
		\node[] (w3) at (3, 5) {}; 
		\node[] (w4) at (6, 3) {};
		\node[] (w5) at (6, 1) {}; 
		\node[] (w6) at (4, 1) {};      
			
		\node[] (q1) at (2.25, 3) {};
		\node[] (q2) at (3.5, 3) {}; 
		\node[] (q3) at (4.75, 3){};        
        
    \end{scope}
    
    \begin{scope}[every edge/.style={draw=black}]
           			
		\path[] (u1) edge[] node {} (u2);  
		\path[] (u2) edge[] node {} (u3);  
		\path[] (u3) edge[] node {} (u1);
           	
		\path[] (u3) edge[] node {} (u4);    
		\path[] (u4) edge[] node {} (u5);  
		\path[] (u5) edge[] node {} (u3);  
           	
		\path[] (v1) edge[] node {} (v2);  
		\path[] (v2) edge[] node {} (v3);  
		\path[] (v3) edge[] node {} (v4);
		\path[] (v4) edge[] node {} (v1);
           	
		\path[] (v3) edge[] node {} (v5);
           	
		\path[] (v5) edge[] node {} (v6);
		\path[] (v6) edge[] node {} (v7);
		\path[] (v7) edge[] node {} (v8);
       	\path[] (v8) edge[] node {} (v9);
        	\path[] (v9) edge[] node {} (v5);
           	
        	\path[] (w1) edge[] node {} (w2);
        	\path[] (w2) edge[] node {} (w3);
       	\path[] (w3) edge[] node {} (w1);       	
           	
        	\path[] (w4) edge[] node {} (w5);
        	\path[] (w5) edge[] node {} (w6);
        	\path[] (w6) edge[] node {} (w4);
           	
        	\path[] (w1) edge[] node {} (w4); 
         	            	 
	\end{scope}
	 
	\begin{scope}[every edge/.style={draw=blue}]	         	            	 
         	            	 
		\path[loosely dashdotted] (u3) edge[] node{} (v2); 
        	\path[loosely dashdotted] (u5) edge[] node{} (w1); 
        	\path[loosely dashdotted] (p2) edge[] node{} (q2);
        	
        	\path[loosely dashdotted] (u2) edge[] node{} (v1);
        	\path[loosely dashdotted] (u4) edge[] node{} (w2); 
        	\path[loosely dashdotted] (w3) edge[] node{} (w4); 
        	\path[loosely dashdotted] (w6) edge[] node{} (w1);
       	\path[loosely dashdotted] (w6) edge[] node{} (v6);

	\end{scope}
        
    \begin{scope}[every node/.style={draw=none,rectangle}]

		\node (Rlabel) at (-2, -2) {$R$};
		\node (C0label) at (1, -3.2) {$C_0$}; 
		\node (C1label) at (-5, 3) {$C_1$}; 		
		\node (C2label) at (5, 5) {$C_2$}; 		

	\end{scope}
    \end{tikzpicture}
    \label{Bridge-Cov-Illa}
    \end{subfigure}    
    \hspace*{\fill}
	\begin{subfigure}{0.49\textwidth}
	\centering
    \begin{tikzpicture}[scale=0.52]
    	\begin{scope}[every node/.style=
    				 {circle, fill=black, draw, inner sep=0pt, 
    				  minimum size = 0.15cm}]

		\node[] (u1) at (-7, 2) {}; 
    		\node[] (u2) at (-5, 0) {}; 
		\node[] (u3) at (-4, 2) {}; 
    		\node[] (u4) at (-4, 5) {}; 
    		\node[] (u5) at (-2, 3) {}; 
            
    		\node[] (v1) at (-3, -2) {}; 
    		\node[] (v2) at (-2, -1) {}; 
    		\node[label={[label distance=1]270:$r$}]  (v3) at (-1, -2) {}; 
    		\node[] (v4) at (-2, -3) {}; 
            
    		\node[label={[label distance=1]270:$u$}] (p1) at (0.33, -2) {}; 
    		\node[] (p2) at (1.66, -2){}; 
            
    		\node[] (v5) at (3, -2) {}; 
		\node[] (v6) at (4, -1) {}; 
		\node[] (v7) at (6, -1) {}; 
		\node[] (v8) at (6, -3) {}; 
		\node[] (v9) at (4, -3) {}; 
			
		\node[] (w1) at (1, 3) {}; 
		\node[] (w2) at (1, 5) {}; 
		\node[] (w3) at (3, 5) {}; 
		\node[] (w4) at (6, 3) {};
		\node[] (w5) at (6, 1) {}; 
		\node[] (w6) at (4, 1) {};      
			
		\node[] (q1) at (2.25, 3) {};
		\node[] (q2) at (3.5, 3) {}; 
		\node[] (q3) at (4.75, 3){};        
        
    \end{scope}
    
    \begin{scope}[every edge/.style={draw=black}]
           			
		\path[] (u1) edge[] node {} (u2);  
		\path[] (u2) edge[] node {} (u3);  
		\path[] (u3) edge[] node {} (u1);
           	
		\path[] (u3) edge[] node {} (u4);    
		\path[] (u4) edge[] node {} (u5);  
		\path[] (u5) edge[] node {} (u3);  
           	
		\path[] (v1) edge[] node {} (v2);  
		\path[] (v2) edge[] node {} (v3);  
		\path[] (v3) edge[] node {} (v4);
		\path[] (v4) edge[] node {} (v1);
           	
		\path[] (v3) edge[] node {} (v5);
           	
		\path[] (v5) edge[] node {} (v6);
		\path[] (v6) edge[] node {} (v7);
		\path[] (v7) edge[] node {} (v8);
       	\path[] (v8) edge[] node {} (v9);
        	\path[] (v9) edge[] node {} (v5);
           	
        	\path[] (w1) edge[] node {} (w2);
        	\path[] (w2) edge[] node {} (w3);
       	\path[] (w3) edge[] node {} (w1);       	
           	
        	\path[] (w4) edge[] node {} (w5);
        	\path[] (w5) edge[] node {} (w6);
        	\path[] (w6) edge[] node {} (w4);
        	\path[] (w1) edge[] node {} (w4); 
         	            	  
        	\path[ultra thick] (u3) edge[] node[right]{$f_1$} (v2); 
        	\path[ultra thick] (u5) edge[] node[above]{$f_2$} (w1); 
        	\path[ultra thick] (p2) edge[] node[left]{$f_3$} (q2);
            
	\end{scope}

	\begin{scope}[every edge/.style={draw=blue}]
		
		\path[loosely dashdotted] (u2) edge[] node{} (v1);
        	\path[loosely dashdotted] (u4) edge[] node{} (w2); 
        	\path[loosely dashdotted] (w3) edge[] node{} (w4); 
        	\path[loosely dashdotted] (w6) edge[] node{} (w1);
       	\path[loosely dashdotted] (w6) edge[] node{} (v6);
       	
	\end{scope}
        
    \begin{scope}[every node/.style={draw=none,rectangle}]
           	
		\node (Rlabel) at (-2, -2) {$R$};
		\node (C0label) at (1, -3.2) {$C_0$}; 
		\node (C1label) at (-5, 3) {$C_1$}; 		
		\node (C2label) at (5, 5) {$C_2$}; 		
		
	\end{scope}
    \end{tikzpicture}
    \label{Bridge-Cov-Illb}
    \end{subfigure}
    \caption{
	Illustration of an iteration of our bridge-covering step. 
	Solid lines indicate edges of the graph $H$, and
	(blue) dash-dotted lines indicate edges of $E(G) - E(H)$. 
	The pseudo-ear $R,f_1,C_1,f_2,C_2,f_3$ covers
	the bridge $ru$ of $C_0$ (right subfigure).
	Thick lines indicate the edges $f_1,f_2,f_3$ of the pseudo-ear. 
	}
    \label{f:Bridge-Cov-Ill}
\end{figure}
}

A brief overview of the bridge~covering step follows: The goal is
to add ``new'' edges to $H$ to obtain a bridgeless 2-edge~cover,
and to pay for these ``new" edges from credits available in $H$
while preserving a credit invariant (stated below).  In each
iteration, we pick a connected component $C_0$ of $H$ such that
$C_0$ has a bridge, then we pick any pendant 2ec-block $R$ of $C_0$,
then we add a set of edges $\{f_1,\dots,f_k\}\subseteq E(G)-E(H)$
that ``covers" the unique bridge of $C_0$ incident to $R$ (possibly,
$k=1$).  Informally speaking, this step merges $k-1$ connected
components $C_1,C_2,\dots,C_{k-1}$ of $H$ with $C_0$ (see the
discussion below).  Each connected component of $H$ has one unit
of so-called \ccredit\ (by the credit invariant stated below), and
we take this credit from each of $C_1,C_2,\dots,C_{k-1}$ and use
that to pay for $k-1$ of the newly added edges.  The challenge is
to find one more unit of credit (since we added $k$ edges), and
this is the focus of our analysis given below.  See
Figure~\ref{f:Bridge-Cov-Ill}.

A detailed discussion of an iteration is presented in
Section~\ref{s:earaugment} below, after we define the notion of a
pseudo-ear; we refer to an iteration (of bridge~covering) as a
pseudo-ear augmentation.

Now, we start on the formal presentation and analysis.
By \cite[Section~5.1,Proposition~5.20]{cdgkn:map},
we may assume without loss of generality that \DTWO\ has the following properties:

\begin{tabular}{lp{0.75\textwidth}}
\label{p:propdtwo}
(*)\qquad & \DTWO\ contains all the zero-edges.
	Every pendant 2ec-block of \DTWO\ that is incident to a zero-bridge
	is a large 2ec-block.
\end{tabular}

Recall that $H$ denotes the current graph, and initially, $H=\DTWO$.
We call a node $v$ of $H$ a \textit{white} node if $v$ belongs to
a 2ec-block of $H$, otherwise, we call $v$ a \textit{black} node.
Observe that all edges of $H$ incident to a black node $v$ are
bridges of $H$, and $v$ is incident to $\ge2$ bridges of $H$.

It is convenient to define the following multi-graphs:
let $\wth$ be the multi-graph obtained from $H$ by contracting each
2ec-block $B_i$ of $H$ into a single node that we will denote by $B_i$
(thus, the notation $B_i$ refers to either a 2ec-block of $H$ or a node of $\wth$).
Observe that each connected component of $\wth$ is a tree (possibly,
an isolated node).  Similarly, let $\wtg$ be the multi-graph obtained
from $G$ by contracting each 2ec-block $B_i$ of $H$ into a single
node.

We call a node $v$ of the multigraph $\wth$ black
if it is the image of a black node of $H$,
otherwise, we call $v$ a white node.
Each 2ec-block of $H$ maps to a white node of $\wth$.
Each bridge of $H$ maps to a bridge of $\wth$.
Clearly, each black node of $\wth$ is incident to $\ge2$ bridges of $\wth$.

\subsection{Credit invariant} \label{s:credits}
We re-assign the credits of \DTWO\ such that the following credit
invariant holds for $H$ at the start/end of every iteration in the
bridge~covering step.
(Note that the credit invariant may ``break" inside an iteration,
while the algorithm is updating information, but this is not relevant
for our correctness proofs.)

For a black node $v$ of $H$, we use $\deg_H^{(1)}(v)$ to denote the
number of unit-bridges incident to $v$ in $H$.

\medskip
\noi
\fbox{ \begin{minipage}{0.9\textwidth}
\begin{itemize}
\item[] \textbf{Credit invariant for $H$}:

\item[(a)] each connected component {is assigned} at least one credit (called \ccredit);

\item[(b)] each connected component that is a small 2ec-block {is assigned}
$\frac13$ credits (called \bcredit);

\item[(c)] every other 2ec-block {is assigned} at least one credit (called \bcredit);

\item[(d)] each black node $v$ {is assigned} $\frac13 \deg_H^{(1)}(v)$ credits
(called \ncredit).
\end{itemize}
\end{minipage}
}
\medskip

Note that the four types of credit are distinct, and
the invariant gives lower bounds.
For example, a connected component that is a large 2ec-block has
one \ccredit\ and at least one \bcredit.

\begin{lemma} \label{lem:startcredit}
The initial credits of \DTWO\ can be re-assigned such that
(the initial) $H=\DTWO$ satisfies the credit invariant.
\end{lemma}
\begin{proof}
Each 2ec-block $B$ of \DTWO\ has $\frac23 \cost(B)$ credits;
in particular, a small 2ec-block has $\frac43$ credits,
and a large 2ec-block has $\ge2$ credits.
Each unit-bridge of \DTWO\ starts with $\frac23$ credits, and it
assigns $\frac13$ credits to each of its end~nodes.
{
The assignment of these credits to black nodes immediately satisfies
part~(d) of the credit~invariant.
However, this also assigns $\frac13$ credits to white end~nodes
of unit edges that we may use below.
}

Next, consider parts~(a), (b), (c) of the credit~invariant, i.e.,
the \ccredit{}s and the \bcredit{}s.
For each bridgeless connected component $C$ of $H$, we split its
credit of $\frac23\cost(C)$ among the \ccredit\ and the \bcredit,
keeping one unit for the \ccredit.

Now, consider any other connected
component $C$ of $H$.  If $C$ contains a large 2ec-block $B$, then
$B$ has $\ge2$ credits, and we take one unit of this credit for the
\ccredit\ of $C$ and leave the remaining credits as the \bcredit\
of $B$.  Otherwise, $C$ contains only small 2ec-blocks, and each
has $\frac43$ credits. If $C$ has at least three 2ec-blocks, then
we take $\frac13$ credits from three of its 2ec-blocks and keep that
as the \ccredit\ of $C$, while leaving $\ge1$ \bcredit\ with each
2ec-block.  If $C$ has exactly two (small) 2ec-blocks $B_1,B_2$,
then note that each is a pendant block, so by property~$(*)$ of
\DTWO\ (see page~\pageref{p:propdtwo}), each of $B_1,B_2$ is incident
to a unit-bridge of \DTWO, and moreover, the (white) end~node of the
unit-bridge in $B_i$ ($i\in[2]$) has $\frac13$ (newly assigned)
credits; thus, $B_1\cup B_2$ has $\frac{10}{3}$ credits, and we
take one credit for the \ccredit\ of $C$ while leaving $\ge1$
\bcredit\ with each of $B_1,B_2$.  Hence, $H$ satisfies parts~(a),
(b), (c), (d) of the credit invariant.
\end{proof}

\subsection{Analysis of a pseudo-ear augmentation} \label{s:earaugment}

In this subsection, our goal is to show that a so-called pseudo-ear
augmentation can be applied to $H$ whenever a connected component
of $H$ has a bridge, such that the cost of the newly added unit-edges
is paid from the credits released by the pseudo-ear augmentation,
and moreover, the credit invariant is preserved.

In the graph $H$, let $C_0$ be a connected component that has a
bridge, let $R$ be a pendant 2ec-block of $C_0$, and let $ru$ be
the unique bridge (of $C_0$) incident to $R$, where $r\in V(R)$.
See Figure~\ref{f:Bridge-Cov-Ill} for an illustration of the following
definition.

\begin{defn}
A \textit{pseudo-ear} of $H$ w.r.t.~$C_0$ starting at $R$ is a
sequence $R,f_1,C_1,f_2,C_2,$ $\dots,$ $f_{k-1},C_{k-1},f_k$, where
$C_0,C_1,\dots,C_{k-1}$ are distinct connected components of $H$,
$f_1,\dots,f_k\in E(G)-E(H)$, each $f_i$, $i\in[k-1]$, has one
end~node in $C_{i-1}$ and the other end~node in $C_i$, $f_1$ has
an end~node in $R$, and $f_k$ has one end~node in $C_{k-1}$ and one
end~node in $C_0-V(R)$.
The end~node of $f_k$ in $C_0-V(R)$ is called the \textit{head~node}
of the pseudo-ear.

Any shortest (w.r.t.\ the number of edges) path of $C_0$ between
$r$ and the head~node of the pseudo-ear is called the \textit{witness~path}
of the pseudo-ear.
\end{defn}

Our plan is to find a pseudo-ear (as above) such that for any
witness~path $Q$, there is at least one unit of credit in $Q-r$.
Let $R^{new}$ denote the 2ec-block that results from the addition
of the pseudo-ear; thus, $R^{new}$ contains $R\cup{Q}$. The \bcredit\
of $R$ is transferred to $R^{new}$; thus, $R^{new}$ satisfies
part~(c) of the credit invariant; see Proposition~\ref{propo:bccredits} below.
After we add the pseudo-ear to $H$, the credits of $Q-r$ are released
(they are no longer needed for preserving the credit invariant,
because $Q\cup{R}$ is merged into $R^{new}$). Informally speaking,
we use the credits released from $Q-r$ to pay for the cost of the
last unit-edge added by the pseudo-ear augmentation.

In the graph $\wtg$, let $\wtc_0$ denote the tree corresponding to
$C_0$ and let $\wtr$ denote the leaf of $\wtc_0$ corresponding to
$R$.  Let $\wtp$ be a shortest (w.r.t.\ the number of edges) path
of $\wtg - E(\wtc_0)$ that has one end~node at $\wtr$ and the other end~node
at another node of $\wtc_0$.  Then $\wtp$ corresponds to a pseudo-ear
$R,f_1,C_1,\dots,C_{k-1},f_k$; the sequence of edges of $E(\wtg)-E(\wth)$
of $\wtp$ corresponds to $f_1,\dots,f_k$ and the sequence of trees
$\wtc_1,\dots,\wtc_{k-1}$ of $\wtp$ corresponds to $C_1,\dots,C_{k-1}$.

It is easy to find a pseudo-ear such that any witness~path $Q$ has
$\ge2$ edges. To see this, observe that $G-u$ is connected (since
$G$ is 2NC); let $P$ be a shortest (w.r.t.\ the number of edges)
path between $R$ and $C_0-V(R)$ in $G-u$; then $P$ corresponds to
our desired pseudo-ear, and the head~node is the end~node of $P$
in $C_0-u-V(R)$. Clearly, any path of $C_0$ between $r$ and the
head~node has $\ge2$ edges, hence, any witness~path of the pseudo-ear
has $\ge2$ edges.

In each iteration (of bridge~covering), we compute a pseudo-ear
using a polynomial-time algorithm that is presented in
the proof of Proposition~\ref{propo:find-pseudo-ear}, see below.

The next lemma is used to lower bound the credit of a witness~path.

\begin{lemma} \label{lem:wpcredit}
Let $\Psi$ be a pseudo-ear of $H$ w.r.t.~$C_0$ starting at $R$,
let $Q$ be a witness~path of $\Psi$,
and let $ru$ be unique bridge of $C_0$ incident to $R$.
Suppose that $Q$ satisfies one of the following:
\begin{itemize}
\item[(a)] $Q$ contains a white node distinct from $r$, or
\item[(b)] $Q$ contains exactly one white node and $\ge3$ bridges, or
\item[(c)] $Q$ contains exactly one white node, exactly two bridges,
	and a black node $v$ such that $\deg_H^{(1)}(v)\ge2$.
\end{itemize}
Then $Q-r$ has at least one credit, and that credit is not needed
for the credit invariant of the graph resulting from the pseudo-ear
augmentation that adds $\Psi$ to $H$.
\end{lemma}
\begin{proof}
First, suppose $Q$ contains a white node $w$, $w\neq{r}$; then, the
2ec-block $B_w$ of $C_0$ that contains $w$ has $\ge1$ \bcredit, and
this credit can be released since $B_w \subsetneq R^{new}$.  Otherwise,
suppose that $Q$ has $\ge3$ bridges; then $Q-r$ has $\ge3$ black
nodes, and each black node is incident to at least one unit-bridge,
and so has $\ge\frac13$ \ncredit{}s; thus, $Q-r$ has $\ge1$ \ncredit,
and this credit can be released since $Q\subsetneq R^{new}$.
Otherwise, suppose that $Q$ has exactly two bridges, and one of the
black nodes $v$ in $Q-r$ has $\deg_H^{(1)}(v)\ge2$; then, $v$ has
$\ge\frac23$ \ncredit{}s; there is another black node in $Q-r$ and
that black node has $\ge\frac13$ \ncredit{}s; thus, $Q-r$ has $\ge1$
\ncredit, and this credit can be released since $Q\subsetneq R^{new}$.
\end{proof}

\begin{proposition} \label{propo:find-pseudo-ear}
There is a polynomial-time algorithm for finding a pseudo-ear (of
$H$ w.r.t.~$C_0$ starting at $R$ ) such that any witness~path $Q$
of the pseudo-ear satisfies one of the three conditions of
Lemma~\ref{lem:wpcredit}.
\end{proposition}
\begin{proof}
We use some simple case analysis to construct a set of nodes $Z$
of $C_0-V(R)$ with $|Z|\leq2$ such that $G-Z$ is connected
and $C_0-V(R)-Z$ is nonempty.
Then there exists a pseudo-ear $\Psi$ with head~node in $C_0-V(R)-Z$,
and it can be found in polynomial time 
by computing a shortest (w.r.t.\ the number of edges) path
in the graph $G-E(C_0)-Z$
between $R$ and $V(C_0)-V(R)-Z$.
Our construction of $Z$ ensures that any witness~path of $\Psi$
satisfies one of the conditions of Lemma~\ref{lem:wpcredit}.

Let $ru$ be the unique bridge of $C_0$ incident to $R$.
Note that $C_0$ has another pendant 2ec-block besides $R$,
and each pendant 2ec-block of $H$ has $\ge3$ nodes, hence,
$C_0-V(R)-Z$ is nonempty for any node~set $Z$ with $|Z|\leq2$.

\begin{enumerate}[(a)]
\item
Suppose $u$ is a white node.
Then choose $Z:=\emptyset$.
For any pseudo-ear and any of its witness~paths $Q$, condition~(a)
of Lemma~\ref{lem:wpcredit} holds, since $Q$ contains $u$.

\item
Suppose $u$ is a black node and $\deg_H^{(1)}(u)\ge2$.
Then we choose $Z:=\{u\}$.
$G-Z$ is connected (since $G$ is 2NC), so there exists a pseudo-ear
with head~node in $C_0-V(R)-Z$, and any of its witness~paths satisfies
condition~(b) or condition~(c) of Lemma~\ref{lem:wpcredit}.

\item
Otherwise, $u$ is a black node and $\deg_H^{(1)}(u)=1$.
In this case, $\deg_H(u)=2$.
Let $w\neq{r}$ be the other neighbour of $u$ in $H$.

\begin{enumerate}[(1)]
\item
Suppose $w$ is a white node, or $w$ is a black node and $\deg_H^{(1)}(w)\ge2$,
Then we choose $Z:=\{u\}$.
$G-Z$ is connected (since $G$ is 2NC), so there exists a pseudo-ear
with head~node in $C_0-V(R)-Z$, and any of its witness~paths satisfies
condition~(a) or condition~(b) or condition~(c) of Lemma~\ref{lem:wpcredit}.

\item
Otherwise, $w$ is a black node and $\deg_H^{(1)}(w)=1$.

In this case, $\deg_H(w)=2$.
Let $x$ denote the other neighbour of $w$ (so, $x\not=u$).
We choose $Z:=\{u,w\}$.
(Note that $\Gamma_H(Z)=\{r,x\}$, so $C_0-Z$ has two connected
components, one contains $r$ and the other one contains $x$.)
Below, we show that $G-Z$ is has a path between $r$ and $x$ by using the fact that $G$
has no \zsplit\ and $G$ has no \usplit.
Hence, there exists a pseudo-ear with head~node in $C_0-V(R)-Z$,
and any of its witness~paths satisfies condition~(b) of Lemma~\ref{lem:wpcredit}.

\begin{enumerate}[(i)]
\item
Suppose $uw$ is a zero-edge.
If $G-Z$ is disconnected, then $uw$ would form a \zsplit;
but, $G$ is an instance of \wsMAP{} and it has no \zsplit. 
Hence, $G-Z$ is connected in this case.

\item
Suppose $uw$ is a unit-edge.
Then, $ru$ is a zero-bridge of $H$, and $wx$ (the other bridge incident
to $w$) is a zero-bridge of $H$.

Suppose $G-Z$ is disconnected and $r$ and $x$ are in different connected components of $G-Z$.
Then we claim that $uw$ would form a \usplit\ (this is verified below).
Since $G$ is an instance of \wsMAP{}, it has no \usplit. 
Therefore, $G-Z$ has a path between $r$ and $x$, hence,
there exists a pseudo-ear with head~node in $C_0-V(R)-Z$.

To verify the claim, consider the graph $G/\{u,w\}$ and let $\hv$
denote the contracted node. $G/\{u,w\}$ has a 2ec-$\hv$-block $B_1$
that contains the zero-edge $r\hv$ and has $\opt(B_1)\ge3$, and
$G/\{u,w\}$ has another 2ec-$\hv$-block $B_2$ that contains the
zero-edge $x\hv$ and has $\opt(B_2)\ge3$.
\big(Remark: For $i\in[2]$, observe that $B_i$ has $\ge4$ nodes;
if $B_i$ has $\ge5$ nodes then $\opt(B_i)\ge3$; if $B_i$ has 4~nodes,
then $B_i$ contains a pendant 2ec-block $B_{i,0}$ of $C_0$ that is
incident to a zero-bridge of $C_0$; observe that $B_{i,0}$ has
3~nodes and has $\ge3$ unit-edges by property~$(*)$ of \DTWO\ (see
page~\pageref{p:propdtwo}),
therefore, $B_i$ is a 2EC graph on four nodes with exactly one zero-edge,
hence, $\opt(B_i)\ge3$.\big)
\end{enumerate}
\end{enumerate}
\end{enumerate}
\end{proof}

\begin{proposition} \label{propo:bccredits}
Suppose that $H$ satisfies the credit~invariant, and
a pseudo-ear augmentation is applied to $H$.
Then the resulting graph $H^{new}$ satisfies the credit invariant.
\end{proposition}
\begin{proof}
We use the notation given above (including $C_0, R, ru, r$).
Let $R,f_1,C_1,f_2,C_2,\dots,f_k$ be the pseudo-ear used in an
iteration, let $v$ be the head~node, and let $Q$ be a witness~path.
Let $R^{new}$ denote the 2ec-block of $H^{new}$ that contains $R$.

For each of the connected components $C_i, i\in[k-1],$ let $s_i$
denote the end~node of $f_i$ in $C_i$, let $t_i$ denote the end~node
of $f_{i+1}$ in $C_i$ (possibly, $s_i=t_i$), and let $P_i$ denote
a shortest (w.r.t.\ the number of edges) path of $C_i$ between $s_i$
and $t_i$. Let $P_0\supseteq Q$ be a path of $C_0$ between $v$ and
the end~node of $f_1$ in $R$.
Let $\augC$ be the cycle $P_0,f_1,P_1,\dots,P_{k-1},f_k$.  Observe
that $\augC$, as well as every 2ec-block of $H$ incident to $\augC$,
is merged into $R^{new}$.

As mentioned above, the \bcredit\ of $R$ is taken to be the \bcredit\
of $R^{new}$; the \ccredit{}s of $C_1,\dots,C_{k-1}$ and the credit
of $Q-r$ are used to pay for $f_1,\dots,f_k$.  All other credits
stay the same. It can be verified that the credit~invariant holds
for $H^{new}$.
\end{proof}

\begin{proof} (of Proposition~\ref{propo:bridgecover})
The proof follows from Lemmas~\ref{lem:startcredit},~\ref{lem:wpcredit},
and Propositions~\ref{propo:find-pseudo-ear},~\ref{propo:bccredits},
and the preceding discussion.

Each iteration, i.e., each pseudo-ear augmentation, can be implemented
in polynomial time, and the number of iterations is $\leq |E(\DTWO)|$.

At the termination of bridge~covering, each connected component of
$H$ is a 2ec-block that has one \ccredit\ and either one \bcredit,
or (in the case of a small 2ec-block) $\frac13$ \bcredit{}s.  By
summing the two types of credit, it follows that each small 2ec-block
has $\frac43$ credits and each large 2ec-block has $\ge2$ credits.
\end{proof}
}

\section{ \label{s:gluing}  The gluing step}
{
In this section, we focus on the gluing step, and we assume that
the input is an instance of \wsMAP{}.
For notational convenience, we denote the input by $G$.
Recall that $G$ is a simple, 2NC graph on $\ge \Nmin$ nodes,
and $G$ has no \zsplit, no \usplit, no \csplit, no \redcycle, and no \redgadget.
(In this section, we use all the properties of $G$ except the
absence of \usplit{}s.)

There are important differences between our gluing step and the
gluing step of \cite{cdgkn:map}. Our gluing step (and overall
algorithm) beats the $\frac74$ approximation threshold because our
pre-processing step eliminates the \csplit\ obstruction and the
\redgadget\ obstruction (these obstructions are not relevant to
other parts of our algorithm).
In an appendix (see Section~\ref{s:appendix}), we present instances
$G$ of MAP that contain \csplit{}s (respectively, \redgadget{}s)
and contain none of the other six obstructions such that
$\opt(G)/\cost(\DTWO(G))\approx\frac74$; informally speaking, our
gluing step, applied to an instance $G$ of \wsMAP{}, finds a 2-ECSS
of cost $\leq\frac53\cost(\DTWO(G))$, but this property need not
hold for other instances of MAP (that are not ``well structured").

We use $H$ to denote the current graph of the gluing~step.  At the
start of the gluing~step, $H$ is a simple, bridgeless graph of
minimum degree two; thus, each connected component of $H$ is 2EC;
clearly, the 2ec-blocks of $H$ correspond to the connected components of $H$.
Recall that a 2ec-block of $H$ is called small if it has $\le2$
unit-edges, and is called large otherwise.
Observe that a small 2ec-block of $H$ is either a 3-cycle with one
zero-edge and two unit-edges, or a 4-cycle with alternating zero-edges
and unit-edges.

The following result summarizes this section:

\begin{proposition} \label{propo:gluing}
At the termination of the bridge-covering step,
let $H$ denote the bridgeless 2-edge~cover computed by the algorithm
and suppose that 
each small 2ec-block of $H$ has $\frac43$ credits and
each large 2ec-block of $H$ has $\ge2$ credits.
Let $\gamma$ denote $\credit(H)$.
Assume that $H$ contains all zero-edges.
Then the gluing step augments $H$ to a 2-ECSS $H'$ of $G$
(by adding edges and deleting edges)
such that $\cost(H') \leq \cost(H) + \gamma - 2$.
The gluing~step can be implemented in polynomial time.
\end{proposition}

Our gluing step applies a number of iterations.
Each iteration picks two or more 2ec-blocks of $H$, and merges them
into a new large 2ec-block by adding some unit-edges and possibly
deleting some unit-edges
such that the following invariant holds for $H$ at the
start/end of every iteration of the gluing step.

\medskip
\noi
\fbox{ \begin{minipage}{0.9\textwidth}
\begin{itemize}
  \item[] \textbf{Invariants for the gluing step}:

\item[$\bullet$]
$H$ is a simple, bridgeless graph of minimum degree two
(hence, the 2ec-blocks of $H$ correspond to the connected components of $H$);

\item[$\bullet$]
(credit invariant)
each small 2ec-block of $H$ has $\frac43$ credits and
each large 2ec-block of $H$ has $\ge2$ credits.

\end{itemize}
\end{minipage}
}
\medskip

It is convenient to define the following multi-graph:
let $\wtg$ be the multi-graph obtained from $G$ by contracting each 2ec-block
$B_i$ of $H$ into a single node that we will denote by $B_i$
(thus, the notation $B_i$ refers to either a 2ec-block of $H$ or a node of $\wtg$).
Observe that $\wtg$ is 2EC.
We call a node of $\wtg$ small (respectively, large) if the corresponding
2ec-block of $H$ is small (respectively, large).
The gluing~step ``operates'' on $G$ and never refers to $\wtg$;
but, for our discussions and analysis, it is convenient to refer to $\wtg$.
(Note that $\wtg$ changes in each iteration, since the current graph
$H$ changes in each iteration.)

Suppose that $\wtg$ has $\ge2$ nodes and has no small nodes.
Then, we pick any (large) node $\wtv$ of $\wtg$.
Since $\wtg$ is 2EC, it has a cycle $\wtc$ incident to $\wtv$.
Let $|\wtc|$ denote the number of edges of $\wtc$;
note that $|\wtc|\ge2$.
Our iteration adds to $H$ the unit-edges corresponding to $\wtc$.
The credit available in $H$ for the 2ec-blocks incident to $\wtc$
is $\ge 2|\wtc|$ and the cost of the augmentation is $|\wtc|$;
hence, we have surplus credit of $2|\wtc|-|\wtc|\ge 2$.
The surplus credit is given to the new large 2ec-block.
Clearly, the credit invariant is preserved.

In general, small nodes may be present in $\wtg$.  If we apply the
above scheme and find a cycle $\wtc$ incident only to small nodes with
$|\wtc|\leq5$, then we fail to maintain the credit invariant (since
only $|\wtc|/3$ credits are available for the new large 2ec-block).
Consider a special case when $\wtg$ has a small node $\sba$ that has
a unique neighbour $B$ and $B$ is large; clearly, there are $\ge2$
parallel edges between $\sba$ and $B$.
Below, we show that $\sba$ and $B$ can be merged to form a new large
2ec-block using an augmentation of net cost one, rather
than two, by deleting one or more unit-edges of $\sba$
(see Lemma~\ref{lem:verify2ec});
then we have surplus credit $\geq2$ for the new large 2ec-block.
For example, if $\sba$ is a 3-cycle of $H$, then there exists a unit-edge
$uw$ of $\sba$ such that $G$ has edges $uv_1$ and $wv_2$ where
$v_1,v_2\in{B}$ (see Lemma~\ref{lem:smallblocks:b}); so the
augmentation adds the unit-edges $uv_1$ and $wv_2$ to $H$ and
discards $uw$ from $H$.
Figure~\ref{f:gluing-example} shows the working of the gluing step on an example.

{
\input{Figures/figure-gluing-step.tex}
}

We present key definitions and results on small 2ec-blocks in
Section~\ref{s:smallblocks}.  Our algorithm for the gluing step and
pseudo-code are presented in Section~\ref{s:gluingalg}.

\subsection{ \label{s:smallblocks} Analysis of small 2ec-blocks}
{
In this subsection, we focus on the small 2ec-blocks of $H$ and we
present the definitions and results that underlie our algorithm for
the gluing step.
See Figure~\ref{f:swaps-examples}(a),(b)
for illustrations of the following discussion.
Recall that $G$ has $\ge\Nmin$ nodes.

{
\begin{figure}[htb]
\subcaptionbox{
	\label{f:Swap-Ex1}
	(\textbf{a})
	The 2ec-block $B_1$ has a swappable edge $uv$, and
	the 2ec-block $B_2$ has a swappable pair $\{x, y\}$.
}[0.45\textwidth]{
	\begin{tikzpicture}[scale=0.45]
    	\begin{scope}[every node/.style={circle, fill=black, draw, 
    				  inner sep=0pt, minimum size = 0.15cm}]    
            
		\draw[very thick,dotted] (0,0) circle (2cm);

		\node[draw=none,fill=none] (lt) at (-1,1) {};
		\node[draw=none,fill=none] (lb) at (-1,-1) {};
        	\node[draw=none,fill=none] (rt) at (1,1) {};
        	\node[draw=none,fill=none] (rb) at (1,-1) {};
            
        	\node (u1) at (-8.3, 0) {}; 
        	\node[label={[label distance=1]90:$u$}] (u2) at (-5, 2) {}; 
        	\node[label={[label distance=1]270:$v$}] (u3) at (-5, -2) {}; 
        
        	\node[label={[label distance=1]90:$x$}] (x1) at (7, 2) {}; 
        	\node[] (x2) at (9, 0) {}; 
        	\node[label={[label distance=1]270:$y$}] (x3) at (7, -2) {}; 
        	\node[] (x4) at (5, 0) {}; 
        
	\end{scope}
          
	\begin{scope}[every edge/.style={draw=black}]
           			
		\path[ultra thick] (u1) edge[] node {} (u2);  
		\path[ultra thick] (u2) edge[] node {} (u3); 
		\path[ultra thick, dashed] (u3) edge[] node {} (u1);  
			
		\path[ultra thick, dashed] (x1) edge[] node {} (x2);  
		\path[ultra thick] (x2) edge[] node {} (x3); 
		\path[ultra thick, dashed] (x3) edge[] node {} (x4);  
		\path[ultra thick] (x4) edge[] node {} (x1); 
            
	\end{scope}
	
	\begin{scope}[every edge/.style={draw=blue}]
	
		\path[loosely dashdotted] (u2) edge[bend left = 15] node {} (lt);
		\path[loosely dashdotted] (u3) edge[bend right = 15] node {} (lb);
		\path[loosely dashdotted] (x1) edge[bend right = 15] node {} (rt); 
		\path[loosely dashdotted] (x3) edge[bend left = 15] node {} (rb); 
		
		\path[loosely dashdotted] (x2) edge[] node {} (x4); 
	
	\end{scope}
        
	\begin{scope}[every node/.style={draw=none,rectangle}]
		\node (B0label) at (0,0) {$B_0$};

		\node (B1label) at (-6.0,0) {$B_1$};
		\node (B2label) at  (8.5,-1.5) {$B_2$};

    \end{scope}
    \end{tikzpicture}
}
\hspace*{\fill}
\subcaptionbox{
    \label{f:Swap-Ex2}
	(\textbf{b})
	The 2ec-block $B_1$ has two swappable edges: $uv$ is good, and $vw$ is bad. 
	The swappable pair $\{x,y\}$ of the 2ec-block $B_2$ is good.
}[0.45\textwidth]{
	\begin{tikzpicture}[scale=0.4]
    	\begin{scope}[every node/.style={circle, fill=black, draw, 
    				  inner sep=0pt, minimum size = 0.15cm }]
        
		\node[label={[label distance=1]45:$u$}]  (u) at (-6, 4) {}; 
		\node[label={[label distance=1]45:$v$}] (v) at (-2, 3) {}; 
		\node[label={[label distance=1]225:$w$}]  (w) at (-5, 0) {}; 
            
		\node[] (w1) at (1, 5) {}; 
		\node[] (w2) at (3, 5) {}; 
		\node[] (w3) at (6, 2) {}; 
		\node[] (w4) at (6, 0) {}; 
		\node[] (w5) at (4, 0) {}; 
		\node[] (w6) at (1, 3) {}; 
		
		\node[] (z1) at (0, -2) {}; 
        	\node[label={[label distance=1]45:$x$}] (x) at (2, -4) {}; 
        	\node[] (z2) at (0, -6) {}; 
        	\node[label={[label distance=1]135:$y$}] (y) at (-2, -4) {}; 
        
	\end{scope}
	
	\begin{scope}[every edge/.style={draw=black}]
        		
		\path[ultra thick] (u) edge[] node {} (v); 
		\path[ultra thick] (v) edge[] node {} (w);
		\path[ultra thick, dashed] (w) edge[] node {} (u);  

		\path[ultra thick] (w1) edge[] node {} (w2);  
		\path[ultra thick, dashed] (w2) edge[] node {} (w3);
		\path[ultra thick] (w3) edge[] node {} (w4);
		\path[ultra thick] (w4) edge[] node {} (w5);
		\path[ultra thick, dashed] (w5) edge[] node {} (w6);
		\path[ultra thick] (w6) edge[] node {} (w1);
		
		\path[ultra thick, dashed] (z1) edge[] node {} (x); 
		\path[ultra thick] (x) edge[] node {} (z2);
		\path[ultra thick, dashed] (z2) edge[] node {} (y);  
        	\path[ultra thick] (y) edge[] node {} (z1);	
        	
	\end{scope}
	
	\begin{scope}[every edge/.style={draw=blue}]
	
		\path[loosely dashdotted] (v) edge[] node {} (w6); 
		\path[loosely dashdotted] (w) edge[] node {} (w5); 
		
		\path[loosely dashdotted] (y) edge[bend left = 90] node {} (u); 
		\path[loosely dashdotted] (x) edge[bend right = 60] node {} (w4); 
		
		\path[loosely dashdotted] (z1) edge[] node {} (z2); 

	\end{scope}
	
	\begin{scope}[every node/.style={draw=none,rectangle}]

		\node (B1label) at (-4.2,2.3) {$B_1$};
		\node (B2label) at (2,-5.5) {$B_2$};
		\node (B3label) at (3.35, 2.5) {$B_3$};

    \end{scope}
	
	\end{tikzpicture}
}
\caption{
	\label{f:swaps-examples}
	Illustrations of swappable edges and swappable pairs of small 2ec-blocks. 
	Solid lines indicate unit-edges of $H$, 
	dashed lines indicate zero-edges of $H$, 
	and (blue) dash-dotted lines indicate edges of $E(G) - E(H)$. 
}
\end{figure}
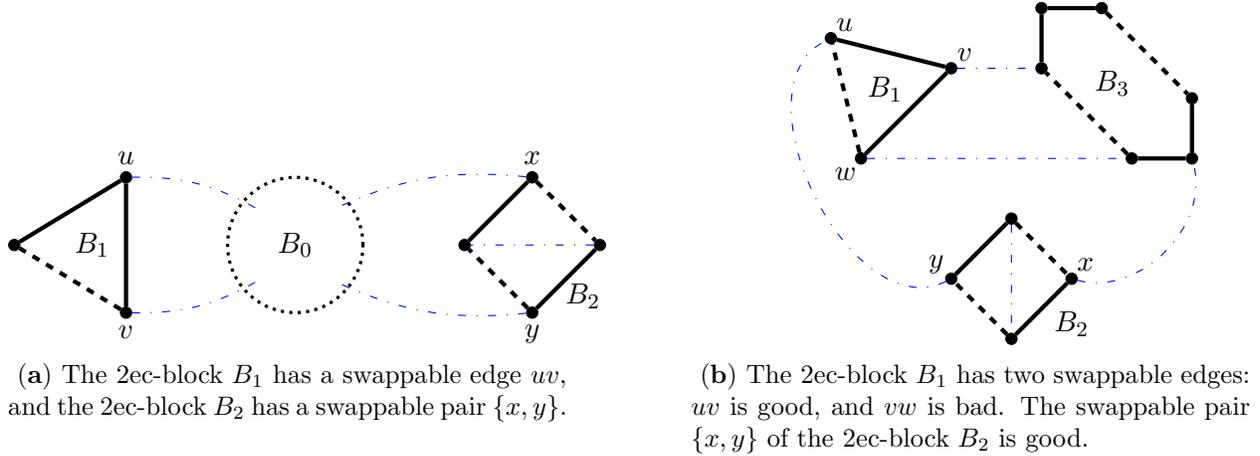
}

\begin{defn} \label{d:sw-edge}
Let $\sba$ be a small 2ec-block of $H$.
A unit-edge $uw$ of $\sba$ is called \textit{swappable}
if both $u$ and $w$ are attachments of $\sba$ in $G$ (that is,
$G$ has an edge $ux$ where $x\in V(G)-\sba$ and
$G$ has an edge $wy$ where $y\in V(G)-\sba$).
\end{defn}

\begin{defn} \label{d:sw-pair}
Let $\sba$ be a small 2ec-block of $H$.
A pair of nodes $\{u,w\}$ of $\sba$ is called a \textit{swappable~pair}
if either
(i)~$uw$ is a swappable~edge of $\sba$, or
(ii)~$u,w$ are not adjacent in $\sba$ (note that $\sba$ is
a 4-cycle in this case) and the other two nodes of $\sba$ are adjacent in $G$
(that is, $E(G)-E(H)$ has a ``diagonal edge'' between the other two nodes of $\sba$).
\end{defn}

\begin{defn} \label{d:sw-goodbad}
Let $\sba$ be a small 2ec-block of $H$.
A swappable~pair $\{u,w\}$ of $\sba$ is called \textit{good}
if there are distinct 2ec-blocks $B_u$ and $B_w$ ($\sba\not=B_u\not=B_w\not=\sba$)
such that
$G$ has an edge $ux$ where $x\in B_u$ and
$G$ has an edge $wy$ where $y\in B_w$;
otherwise, $\{u,w\}$ is called a \textit{bad} swappable~pair of $\sba$.
A good (respectively, bad) swappable~edge of $\sba$ is defined similarly.
\end{defn}

\authremark{~Observe that each iteration merges two or more 2ec-blocks
of $H$ (see the discussion following Proposition~\ref{propo:gluing}).
Consider a small 2ec-block $\sba$ of $H$ that stays unchanged over
several iterations. After one of these iterations, a swappable~pair
$\{u,w\}$ of $\sba$ may change from good to bad, but $\{u,w\}$ cannot
change from bad to good.}

{
\begin{lemma} \label{lem:smallblocks:a}
Let $\sba$ be a small 2ec-block of $H$.
If $\sba$ is adjacent (in $G$) to a unique 2ec-block $B$,
then $B$ is large.
(That is, if there is 2ec-block $B$ such that $\Gamma_G(V(\sba)) \subseteq V(B)$,
then $B$ is large.)
\end{lemma}
\begin{proof}
This follows from the absence of \csplit{}s in $G$.
In more detail, suppose that $B$ is small.
Then, $G-V(B)$ has $\ge2$ connected components, where one connected component
is $\sba$ and another connected component is in the nonempty subgraph $G-V(B)-V(\sba)$.
Then $B$ would satisfy the definition of an \csplit, see Definition~\ref{d:csplit}.
\big(To verify this, note that $|V(G)|\ge\Nmin$, and the cut
$\delta(V(B))$ consists of unit-edges since it is a subset of
$E(G)-E(H)$ (all zero-edges are in $H$); moreover, each connected~component
of $G-V(B)$ contains a 2ec-block (that has $\ge3$ nodes), hence,
$G/B$, with $\hv$ denoting the contracted node, has two (or more)
2ec-$\hv$-blocks $\hB_1, \hB_2$, such that for $i\in[2]$, either
$\hB_i$ has $\ge5$ nodes and so has $\opt(\hB_i)\ge3$, or $\hB_i$
has $4$~nodes and $\le1$ zero-edges and so has $\opt(\hB_i)\ge3$.\big)
Since $G$ has no \csplit{}s, we have a contradiction.
\end{proof}

\begin{lemma} \label{lem:smallblocks:b}
Let $\sba$ be a small 2ec-block of $H$.
Then $\sba$ has at least one swappable pair.
Moreover,
if $\sba$ is a 3-cycle, then
$\sba$ has at least one swappable edge.
\end{lemma}
\begin{proof}
$\sba$ has $\ge2$ attachments, since $G$ is 2NC.
If $\sba$ has $\ge3$ attachments, then $\sba$ has a unit-edge $f$ that
contains two distinct attachments, and $f$ is a swappable edge of
$\sba$. Now, suppose that $\sba$ has exactly 2~attachments $u,w$.
If $f=uw$ is a zero-edge of $\sba$, then $f$ would form a \zsplit\
of $G$, whereas an instance of \wsMAP{} has no \zsplit{}s.
Thus, either $\sba$ has a unit-edge $f$ between the two attachments $u,w$
(clearly, $f$ is a swappable edge of $\sba$),
or $u$ and $w$ are not adjacent in $\sba$ (then, $\sba$ is a 4-cycle).
Consider the latter case.
Let $v_1$ and $v_2$ be the other two nodes of $\sba$.
If $v_1$ and $v_2$ are not adjacent in $G$,
then $\deg_G(v_1)=\deg_G(v_2)=2$.
Then $\sba$ would form an \redcycle\ of $G$, whereas an instance of
\wsMAP{} has no \redcycle{}s.
Otherwise, if $v_1v_2\in E(G)$, then $\{u,w\}$ is a swappable pair of $\sba$.
\end{proof}

\begin{lemma} \label{lem:smallblocks:c}
Let $\sba$ be a small 2ec-block of $H$.
If $\sba$ is a 3-cycle, and $\sba$ is adjacent (in $G$) to at least two
other 2ec-blocks, then it has a good swappable edge.
\end{lemma}
\begin{proof}
$\sba$ has $\ge2$ attachments, since $G$ is 2NC.
Let $u$ be the node of $\sba$ that is incident to two unit-edges of
$\sba$ and let $vw$ be the zero-edge of $\sba$.  Then $u$ must be an
attachment of $\sba$ (otherwise, $vw$ would form a \zsplit\ of $G$).
Since $\sba$ is adjacent (in $G$) to at least two other 2ec-blocks
(and $\sba$ has $\ge2$ attachments),
there exists another attachment of $\sba$, say $w$, and there exist
distinct 2ec-blocks, say $B_u$ and $B_w$, such that
$u$ is adjacent to $B_u$ and
$w$ is adjacent to $B_w$, where $\sba\not=B_u\not=B_w\not=\sba$.
Then $uw$ is a good swappable edge of $\sba$.
\end{proof}
}

Suppose that the current graph $H$ has no good swappable~pairs,
that is, for every small 2ec-block $\sba$ of $H$, every swappable pair
of $\sba$ is bad.
To ``merge away" the remaining small 2ec-blocks of $H$, we construct
the following auxiliary digraph $D^{aux}$: there is a node for each
2ec-block of $H$, and we call the nodes corresponding to the small
2ec-blocks the \textit{red} nodes, and the other nodes the
\textit{green} nodes; for each small 2ec-block $\sba$ of $H$ and each
of its swappable~pairs $\{u,w\}$, $D^{aux}$ has an arc $(\sba,B)$ where
$B$ corresponds to the unique 2ec-block $B$ of $H$ such that
$\Gamma_G(\{u,w\})\subseteq V(B)\cup{V(\sba)}$. Observe that each red
node of $D^{aux}$ has at least one outgoing arc. See Figure~\ref{f:Daux-Ex}.

{
\begin{figure}[htb]
	\centering
	\begin{subfigure}{0.49\textwidth}
	\centering
    \begin{tikzpicture}[scale=0.40]
    	\begin{scope}[every node/.style={circle, fill=black, draw, inner sep=0pt,
            minimum size = 0.15cm
            }]
            
            \node[] (u1) at (-7, 1) {}; 
            \node[] (u2) at (-5, 1) {};
            \node[] (u3) at (-5, -1) {}; 
            \node[] (u4) at (-7, -1) {}; 
            
            \node[] (v1) at (-3, 1) {}; 
            \node[] (v2) at (-1, 1) {};
            \node[] (v3) at (-1, -1) {}; 
            \node[] (v4) at (-3, -1) {}; 
            
            \node[] (x1) at (1, 1) {}; 
            \node[] (x2) at (3, 1) {};
            \node[] (x3) at (3, -1) {}; 
            \node[] (x4) at (1, -1) {}; 
            
            \node[] (w1) at (5, 1){}; 
            \node[] (w2) at (6, 2){}; 
            \node[] (w3) at (7, 1){};
            \node[] (w4) at (7, -1){}; 
            \node[] (w5) at (6, -2){}; 
            \node[] (w6) at (5, -1){}; 
        
        \end{scope}
        \begin{scope}[every edge/.style={draw=black}]
        		
        		\path[ultra thick] (u1) edge[] node {} (u2);
        		\path[ultra thick, dashed] (u2) edge[] node {} (u3);
        		\path[ultra thick] (u3) edge[] node {} (u4); 
        		\path[ultra thick, dashed] (u4) edge[] node {} (u1); 
        		
        		\path[ultra thick, dashed] (v1) edge[] node {} (v2);
        		\path[ultra thick] (v2) edge[] node {} (v3);
        		\path[ultra thick, dashed] (v3) edge[] node {} (v4); 
        		\path[ultra thick] (v4) edge[] node {} (v1); 
        		
			\path[ultra thick, dashed] (x1) edge[] node {} (x2);
        		\path[ultra thick] (x2) edge[] node {} (x3);
        		\path[ultra thick, dashed] (x3) edge[] node {} (x4); 
        		\path[ultra thick] (x4) edge[] node {} (x1);          		
        		
        		\path[ultra thick] (w1) edge[] node {} (w2);  
        		\path[ultra thick, dashed] (w2) edge[] node {} (w3);
        		\path[ultra thick] (w3) edge[] node {} (w4);
        `	\path[ultra thick] (w4) edge[] node {} (w5);
        		\path[ultra thick, dashed] (w5) edge[] node {} (w6);
        		\path[ultra thick] (w6) edge[] node {} (w1);
        		
        	\end{scope}
        	\begin{scope}[every edge/.style={draw=blue}]
        		
        		\path[loosely dashdotted] (u4) edge[bend right = 30] 
        			node {} (x3); 
        		\path[loosely dashdotted] (u3) edge[bend right = 30] 
        			node {} (x3); 
        		\path[loosely dashdotted] (u2) edge[bend left = 20] 
        			node{} (w2); 
        		\path[loosely dashdotted] (x3) edge[bend right = 30] 
        			node{} (w5); 
        			
        		\path[loosely dashdotted] (x1) edge[] node {} (v2); 
        		\path[loosely dashdotted] (x4) edge[] node {} (v2); 
        		\path[loosely dashdotted] (v1) edge[] node {} (u2); 
        		\path[loosely dashdotted] (v4) edge[] node {} (u2); 
            
        \end{scope}
        	\begin{scope}[every node/.style={draw=none,rectangle}]

			\node (B1label) at (2, 0) {$B_1$};
			\node (B2label) at (-2, 0) {$B_2$};
			\node (B3label) at (-6, 0) {$B_3$};
			\node (B0label) at (6, 0) {$B_0$};

    		\end{scope}
        
    \end{tikzpicture}
    \label{Daux-Exa}
    \end{subfigure}
    \hspace*{\fill}
    \begin{subfigure}{0.49\textwidth}
    \centering
       \begin{tikzpicture}[scale=0.52]
       	\begin{scope}[every node/.style={circle, fill=black, draw, inner sep=0pt,
            minimum size = 0.25cm
            }]
            
            \node[label={[label distance=1]90:$B_3$}] (B3) at (-3, 0) {}; 
            \node[label={[label distance=1]90:$B_2$}] (B2) at (-1, 0) {}; 
            \node[label={[label distance=1]90:$B_1$}] (B1) at (1, 0) {}; 
            \node[label={[label distance=1]90:$B_0$}] (B0) at (3, 0) {}; 
         \end{scope}
         
         \draw [->] (B1) edge (B2) (B2) edge (B3) (B3) edge[bend right = 60] (B1);
    		\end{tikzpicture}
    \label{Daux-Exb}
    \end{subfigure}
    \caption{
	An illustration of the auxiliary digraph $D^{aux}$ (right subfigure). 
	Solid lines indicate unit-edges of $H$, 
	dashed lines indicate zero-edges of $H$ 
	and (blue) dash-dotted lines indicate edges of $E(G) - E(H)$. 
	}
    \label{f:Daux-Ex}
\end{figure}
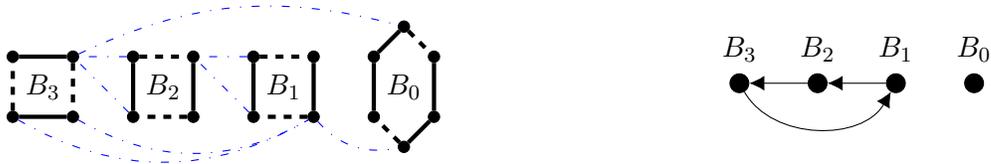
}

{
\begin{lemma} \label{lem:Daux}
Suppose that there exist no good swappable pairs.
Then, $D^{aux}$ does not have a pair of red nodes $\sba_1, \sba_2$
such that $(\sba_1,\sba_2)$ is the unique outgoing arc of $\sba_1$ and
          $(\sba_2,\sba_1)$ is the unique outgoing arc of $\sba_2$
(that is,
if $D^{aux}$ has a directed 2-cycle $C$ on the red nodes,
then one of the red nodes incident to $C$ has $\ge2$ outgoing arcs).
\end{lemma}
\begin{proof}
Suppose that $\sba_1$ and $\sba_2$ are red nodes of $D^{aux}$, and
$\sba_1,\sba_2,\sba_1$ is a directed 2-cycle of $D^{aux}$.
For the sake of contradiction, assume that $D^{aux}$ has exactly
one arc outgoing from each of $\sba_1$ and $\sba_2$.
Clearly, both $\sba_1$ and $\sba_2$ are small 2ec-blocks of the current~graph $H$.
Moreover, $\sba_1$ has a bad swappable pair $\{u_1,w_1\}$ and all
neighbours of $\{u_1,w_1\}$ (in $G$) are in $\sba_1 \cup \sba_2$.

The proof is completed via a few claims and their proofs.
Let $\att$ denote the set of attachments of $(\sba_1 \cup \sba_2)$. 

\begin{claim} \label{cl:a}
$|\att|\ge2$,
$\att\cap{V(\sba_1)}$ is nonempty, and
$\att$ is disjoint from $\{u_1,w_1\}$.
\end{claim}

Since $G$ is 2NC and $G-(\sba_1\cup \sba_2)$ is non-empty, $\att$ has $\geq2$
nodes. Consider $\att\cap{V(\sba_1)}$, the set of attachments of
$(\sba_1\cup{\sba_2})$ that are in $\sba_1$.  By Lemma~\ref{lem:smallblocks:a},
$\att\cap{V(\sba_1)}$ is non-empty (otherwise, all neighbours of $\sba_1$
(in $G$) would be in $\sba_2$, so $\sba_2$ would form an \csplit\ of $G$,
and this can be verified using the arguments in the proof of
Lemma~\ref{lem:smallblocks:a}). 
Observe that $\att$ is disjoint from $\{u_1,w_1\}$;
otherwise, if say $u_1\in \att$, then $G$ has an edge between $u_1$
and a node of $G-(\sba_1\cup \sba_2)$ as well as an edge between $w_1$
and $\sba_2$, hence, $\{u_1,w_1\}$ would be a good swappable pair of
$\sba_1$, and this would contradict the hypothesis of the lemma.
Thus, Claim~\ref{cl:a} is proved.

\begin{claim} \label{cl:b}
$\sba_1$ is a 4-cycle, a unit-edge $e_1$ of $\sba_1$ is a bad swappable
edge such that all neighbours (in $G$) of the two end~nodes of $e_1$
are in $\sba_1\cup{\sba_2}$; moreover, $\att$ contains exactly one node
of $\sba_1$ and that node is not incident to $e_1$.
\end{claim}

The proof of Claim~\ref{cl:b} examines a few cases.
There are two cases for the swappable pair $\{u_1,w_1\}$:
\begin{itemize}
\item[(i)]
$u_1,w_1$ are not adjacent in $\sba_1$, so $\sba_1$ is a 4-cycle,
and the other two nodes of $\sba_1$ are adjacent in $G$ (i.e., $E(G)-E(H)$
has a ``diagonal edge'' between the other two nodes of $\sba_1$), or
\item[(ii)]
$\sba_1$ has a unit-edge between $u_1$ and $w_1$.
\end{itemize}
Consider case~(i).
Let $v_1$ be a node of $(V(\sba_1)-\{u_1,w_1\})\cap\att$.
Then, the unit-edge $f_1$ of $\sba_1$ incident to $v_1$ is a
good swappable edge (because $v_1$ has a neighbour in $G-(\sba_1\cup
\sba_2)$ and the other end~node of $f_1$ has a neighbour in $\sba_2$).
This contradicts the hypothesis of the lemma.
Hence, case~(i) cannot occur.
Now, consider case~(ii).
Clearly, $e_1=u_1w_1$ is a bad swappable edge of $\sba_1$,
and (by Definition~\ref{d:sw-goodbad})
all neighbours (in $G$) of $u_1$ and $w_1$ are in $\sba_1\cup{\sba_2}$.
One possibility is that $\att$ contains exactly one node $v_1$ of
$V(\sba_1)-\{u_1,w_1\}$.
Then $\sba_1$ cannot be a 3-cycle (otherwise, the unit-edge of $\sba_1$
incident to $v_1$ would be a good swappable edge, and this would
contradict the hypothesis of the lemma).
Thus, $\sba_1$ is a 4-cycle such that $\att$ contains exactly one node
of $\sba_1$ and that node is not incident to the unit-edge $e_1=u_1w_1$.
One more case has to be examined to complete the proof of the claim.
Suppose that $\sba_1$ is a 4-cycle, and both nodes of
$V(\sba_1)-\{u_1,w_1\}$ are in $\att$. Then the unit-edge $f_1$ of
$\sba_1$ between those two nodes would be a swappable edge of $\sba_1$.
Clearly, $f_1$ cannot be a good swappable edge, since that would contradict
the hypothesis of the lemma. Hence, $f_1$ would be a bad swappable
edge, and there would exist another 2ec-block $B$ of $H$
($\sba_1\not=B\not=\sba_2$) such that $B\cup{\sba_1}$ contains all neighbours
(in $G$) of the end~nodes of $f_1$. Then, $D^{aux}$ would have the
arc $(\sba_1,B)$, and this contradicts the assumption that $(\sba_1,\sba_2)$
is the unique outgoing arc of $\sba_1$.

Similar properties hold for $\sba_2$ as well; that is,
$\sba_2$ is a 4-cycle, a unit-edge $e_2$ of $\sba_2$ is a bad swappable
edge such that all neighbours (in $G$) of the two end~nodes of
$e_2$ are in $\sba_1\cup{\sba_2}$; moreover, $\att$ contains exactly one node
of $\sba_2$ and that node is not incident to $e_2$.
Clearly, the subgraph of $G$ induced by $\sba_1\cup{\sba_2}$ forms an
\redgadget. This is a contradiction, since $G$ is an instance of
\wsMAP{} so $G$ contains no \redgadget.
\end{proof}
}

By the above lemma, $D^{aux}$ either has an arc $(\sba,B)$ from a red
node $\sba$ to a green node $B$, or it has a directed path $\sba_1,\sba_2,\sba_3$
on three red nodes.
In both cases, we can apply a merge step to obtain a new large 2ec-block
(i.e., a green node) while preserving the credit invariant.  
More details are presented in the next subsection.

}

\subsection{ \label{s:gluingalg} Algorithm for the gluing step}
{
In this subsection, we explain the working of the algorithm for the
gluing step, based on the results in the previous subsection, and
then we present pseudo-code for this algorithm.

Consider any small 2ec-block $\sba$ that has a good swappable~pair
$\{u,w\}$ such that $u$ is adjacent (in $G$) to another 2ec-block
$B_u$, and $w$ is adjacent (in $G$) to another 2ec-block $B_w$, and
$\sba\not=B_u\not=B_w\not=\sba$.  Observe that $G-V(\sba)$ is connected,
otherwise, $\sba$ would be an \csplit\ of $G$ (the arguments in the
proof of Lemma~\ref{lem:smallblocks:a} can be used to verify this
statement).  Hence, $\wtg-\sba$ has a path between $B_u$ and $B_w$;
adding the edges $\sba{}B_u$ and $\sba{}B_w$ to this path gives a cycle $\wtc$
of $\wtg$.
We merge the 2ec-blocks incident to $\wtc$ into a new large 2ec-block
by adding the unit-edges corresponding $\wtc$ to $H$.  Moreover,
if $uw\in{E(\sba)}$, then we discard $uw$ from $H$, otherwise, $\sba$ is
a 4-cycle (with two zero-edges) and $E(G)-E(H)$ has a unit-edge $f$
between the two nodes of $\sba{}-\{u,w\}$, and in this case, we add the
edge $f$ to $H$ and we discard the two unit-edges of $\sba$ from $H$.
The credit available in $H$ for $\wtc$ is $\ge \frac43 |\wtc|$ and
the net cost of the augmentation is $|\wtc|-1$; hence, we have
surplus credit of $\frac13 |\wtc| + 1\ge 2$ (since $|\wtc|\ge3$).
The surplus credit is given to the new large 2ec-block.

The gluing step applies the above iteration until there are no good
swappable~pairs in the current graph $H$.  Then the auxiliary digraph
$D^{aux}$ is constructed.  By Lemma~\ref{lem:Daux}, $D^{aux}$ has
either (i)~an arc $(\sba,B)$ from a red node $\sba$ to a green node $B$,
or (ii)~a directed path $\sba_1,\sba_2,\sba_3$ on three red nodes.

In the first case, $\sba$ is a small 2ec-block, $B$ is a large 2ec-block,
and $\sba$ has a swappable~pair $\{u,w\}$ such that $\sba\cup{B}$ contains
all neighbours (in $G$) of $\{u,w\}$. We merge $\sba$ and
$B$ into a new large 2ec-block as follows.  We add two unit-edges
between $\sba$ and $B$ to $H$ (one edge is incident to $u$ and the
other edge is incident to $w$).  Moreover, if $uw\in{E(\sba)}$, then
we discard $uw$ from $H$, otherwise, $\sba$ is a 4-cycle (with two
zero-edges) and $E(G)-E(H)$ has a unit-edge $f$ between the two
nodes of $\sba{}-\{u,w\}$, and in this case, we add the edge $f$ to $H$
and we discard the two unit-edges of $\sba$ from $H$.
The credit available in $H$ for $\sba\cup{B}$ is $\ge \frac43+2$ and
the net cost of the augmentation is one; hence, we have surplus
credit of $\frac13+2\ge 2$. The surplus credit is given to the new
large 2ec-block.
Consider the second case.
Then $\sba_1,\sba_2,\sba_3$ are small 2ec-blocks such that $\sba_1$ has a
swappable~pair $u_1w_1$ such that $\Gamma_G(\{u_1,w_1\})\subseteq
V(\sba_1)\cup{V(\sba_2)}$, and $\sba_2$ has a swappable~pair $u_2w_2$ such
that $\Gamma_G(\{u_2,w_2\})\subseteq V(\sba_2)\cup{V(\sba_3)}$.  We add
two unit-edges between $\sba_1$ and $\sba_2$ to $H$ (one edge is incident
to $u_1$ and the other edge is incident to $w_1$), and then we
either discard one unit-edge from $H$ (if $u_1w_1\in{E(\sba_1)}$) or
we add another edge to $H$ and discard two unit-edges of $\sba_1$ from
$H$ (if $u_1w_1\not\in{E(\sba_1)}$).
We apply a similar augmentation to $\sba_2$ and $\sba_3$ using the
swappable~pair $\{u_2,w_2\}$.  The credit available in $H$ for
$\sba_1\cup{\sba_2}\cup{\sba_3}$ is $\ge 3(\frac43)=4$ and the net cost of
the augmentation is two; hence, we have surplus credit of $\ge 4-2$.
The surplus credit is given to the new large 2ec-block.

By repeatedly applying the above iteration (that merges red nodes
of $D^{aux}$ into green nodes), we obtain a current graph $H$ that
has no small 2ec-blocks. As discussed above, the merge step is
straightforward when all 2ec-blocks of $H$ are large.

\begin{lemma} \label{lem:verify2ec}
After every merge step, the subgraph $B^{new}$ constructed by that
step (that is a so-called large 2ec-block) is 2EC.
\end{lemma}
\begin{proof}
Adding the edge~set of a cycle $\wtc$ (of $\wtg$) to the current~graph
$H$, call it $H^{prev}$, creates a 2EC subgraph $B^{new}$ that
contains all the 2ec-blocks $B_0,\dots,B_k$ (of $H^{prev}$) incident
to $\wtc$; note that $k$ is $\ge1$.

If the merge step discards a swappable edge $uw$ of say $B_0$, then
note that $\wtc$ contains two edges incident to $B_0$, one
incident to $u$ and one incident to $w$.  
Clearly, the resulting graph $H^{new}$ has two edge-disjoint $u,w$
paths (one is in $B_0$ and the other is in
$E(\wtc)\cup{E(B_1)}\cup\dots\cup{E(B_k)}$), hence, by
Proposition~\ref{propo:2ecdiscard}, $B^{new}$ is 2EC.

Suppose that the merge step applies the augmentation for a swappable
pair $\{u,w\}$ of $B_0$ such that $uw\not\in E(B_0)\subsetneq E(H^{prev})$;
note that $\wtc$ contains two edges incident to $B_0$, one
incident to $u$ and one incident to $w$.  
Clearly, $B_0$ is a 4-cycle; denote it by $u,v_1,w,v_2,u$.
Let $uv_1$ and $wv_2$ be zero-edges.
By the definition of a swappable pair, $v_1v_2\in E(G)-E(H^{prev})$.
Assume that the augmentation first adds $\{v_1v_2\}\cup E(\wtc)$
to $H^{prev}$, and then (sequentially) discards the unit-edges $v_1w$
and $v_2u$. The resulting graph $H^{new}$ has two edge-disjoint
$v_1,w$ paths (one is $v_1,v_2,w$ and the other is in
$E(\wtc)\cup{E(B_1)}\cup\dots\cup{E(B_k)}$), and has two edge-disjoint
$v_2,u$ paths (one is $v_2,v_1,u$ and the other is in
$E(\wtc)\cup{E(B_1)}\cup\dots\cup{E(B_k)}$), hence, by
Proposition~\ref{propo:2ecdiscard}, $B^{new}$ is 2EC.
\end{proof}

\newenvironment{absolutelynopagebreak}
  {\par\nobreak\vfil\penalty0\vfilneg
   \vtop\bgroup}
  {\par\xdef\tpd{\the\prevdepth}\egroup
   \prevdepth=\tpd}
\begin{absolutelynopagebreak}
{
{
\medskip
\noi
\fbox{ \begin{minipage}{\textwidth}

\textbf{Gluing Algorithm (outline)}
\begin{itemize}
\item[(1)]
\textbf{while} there exists a small 2ec-block of $H$
that has a good swappable pair
\begin{quote}
let $\sba$ be a small 2ec-block of $H$ that 
has a good swappable pair $\{u,w\}$;
\\
find a cycle $\wt{C}=\sba,B_1,B_2,\dots,B_k,\sba$ of $\wt{G}$,
where $k\ge2$, $B_1=B_u$, $B_k=B_w$,
$\sba,B_u,B_w$ are distinct and the 2ec-block $B_u$ (respectively,
$B_w$) is adjacent in $G$ to $u$ (respectively, $w$);
then, apply the augmentation that merges  $\sba$,$B_1$,\dots,$B_k$
into a single large 2ec-block using the swappable pair $\{u,w\}$
(such that the credit invariant is preserved);
\end{quote}

\item[(2)]
(every swappable~pair of each remaining small 2ec-block of $H$ is bad)
\begin{itemize}
\item[(2.a)]
construct the auxiliary digraph $D^{aux}$;

\item[(2.b)]
\textbf{while} $D^{aux}$ has a red~node
\begin{quote}
apply a valid augmentation by either merging three red~nodes
using two bad swappable pairs (such that the credit invariant is
preserved), or
merging a red~node with a green~node using a bad swappable pair
(such that the credit invariant is preserved), and then update
$D^{aux}$ appropriately;
\end{quote}
\end{itemize}

\item[(3)]
(every remaining 2ec-block of $H$ is large and has $\ge2$ credits)
\\
\textbf{while} $H$ has two~or~more 2ec-blocks
\begin{quote}
apply a valid augmentation via a cycle $\wt{C}$ of $\wt{G}$
(such that the credit invariant is preserved);
\end{quote}

\item[(4)]
\textbf{stop} ($H$ forms a single large 2ec-block that has $\ge2$ credits).

\end{itemize}
\end{minipage}
}
\medskip
}
}
\end{absolutelynopagebreak}

\medskip
\begin{proof}
(of {Proposition}~\ref{propo:gluing})
The proof follows from
Lemmas~\ref{lem:smallblocks:a},
\ref{lem:smallblocks:b},
\ref{lem:smallblocks:c},
\ref{lem:Daux},
\ref{lem:verify2ec},
and the preceding discussion.
At the termination of the gluing~step,
let $H'$ denote the current~graph;
$H'$ is a 2-ECSS of $G$ and it has $\ge2$ credits;
hence, $\cost(H')$ satisfies the claimed upper~bound.

Each merge~step can be implemented in polynomial time, and the
number of merge~steps is $O(|V(G)|)$, hence, the gluing~step can
be implemented in polynomial time.
\end{proof}

\medskip
\begin{proof} \label{prf:approxbydtwo}
(of {Theorem}~\ref{thm:approxbydtwo})
The proof follows from
{Proposition}~\ref{propo:bridgecover} (on the bridge~covering~step)
and
{Proposition}~\ref{propo:gluing} (on the gluing~step).
These two results imply that the algorithm runs in polynomial~time.

Let $H'$ denote the 2-ECSS of $G'$ computed at the termination of
the gluing~step, and let $H^{(0)}$ denote the current~graph at the
termination of the bridge~covering~step.
By {Propositions}~\ref{propo:bridgecover} and~\ref{propo:gluing},
$\cost(H') \leq \cost(H^{(0)})+\credit(H^{(0)}) - 2
	\leq \cost(\DTWO(G')) + \credit(\DTWO(G')) - 2 = \frac53 \cost(\DTWO(G')) - 2$.
\end{proof}
}

}

\clearpage
\section{ \label{s:appendix}  Appendix: Examples for \csplit\ and \redgadget}
{
In this section, we present instances $G$ of MAP that contain either
the \csplit\ obstruction or the \redgadget\ obstruction (and none
of the other six obstructions) such that
$\opt(G)/\cost(\DTWO(G))\approx\frac74$; each of these instances
has multiple copies of the relevant obstruction; one of these
instances is based on an instance given in \cite[Section~7.1]{cdgkn:map}.
We present another instance on 12~nodes that has one copy of the
\csplit\ obstruction such that our gluing step fails (that is, no
valid augmentation can be applied, see the pseudo-code in
Section~\ref{s:gluing}).

\input{Figures/figures-S34R8.tex}

The graph $G:=G^{(1)}$ of the first instance
has 12~nodes $u_i,v_i,w_i,x_i$ for $i\in[3]$,
six zero-edges $v_iw_i, u_ix_i$ for $i\in[3]$,
and 13~unit-edges:
$u_iv_i, w_ix_i$ for $i\in[3]$,
$v_1x_1$, $u_1u_2$, $w_1w_2$, $u_2x_3$, $v_2x_3$, $w_2v_3$, $u_3w_3$;
see the graph in Figure~\ref{f:S34R8-example}(a).
Let $H=\DTWO(G^{(1)})$ consist of the three 4-cycles of cost two,
$B_i=u_i,v_i,w_i,x_i,u_i$ for $i\in[3]$.
Observe that $B_2$ forms an \csplit{} of $G^{(1)}$.
Although $e=u_2v_2$ is a good swappable edge of the 2ec-block $B_2$
of $H$, there exists no augmenting cycle $\wtc$ in the graph
$\wtg=\widetilde{G^{(1)}}$ that allows $e$ to be discarded.
Although the 2ec-block $B_1$ of $H$
has a bad swappable pair $\{u_1,w_1\}$ and the 2ec-block $B_3$ of $H$
has a bad swappable pair $\{v_3,x_3\}$,
there is no valid augmentation that preserves the credit invariant,
see the pseudo-code of Section~\ref{s:gluing}.

The graph $G:=G_k^{(2)}$ of the second instance consists of
$k$ copies $J_1,\dots,J_k$ of a gadget subgraph $J$,
$B_0=w_1,\dots,w_6,w_1$, which is a 6-cycle of cost~three,
and two unit-edges between each $J_i$ and $B_0$.
The gadget subgraph $J$ consists of 8 nodes $v_1,\dots,v_8$ and 11 edges;
there are four zero-edges $v_1v_4,v_2v_3,v_5v_8,v_6v_7$, and
seven unit-edges $v_1v_2,v_1v_7,v_2v_5,v_3v_4,v_3v_8,v_5v_6,v_7v_8$;
see the subgraph induced by the nodes $v_1,\dots,v_8$ in
Figure~\ref{f:S34R8-example}(b); observe that 8~of the 11~edges form two
disjoint 4-cycles (namely, $v_1,v_2,v_3,v_4,v_1$ and $v_5,v_6,v_7,v_8,v_5$)
and the other three edges are $v_2v_5$, $v_3v_8$, and $v_1v_7$.
$G_k^{(2)}$ has two unit-edges between each copy of the gadget subgraph $J_i$
($i=1,\dots,k$) and $B_0$; these two edges are incident to the nodes
$v_1$ and $v_3$ of $J_i$ (see the illustration in
Figure~\ref{f:S34R8-example}(b)) and to the nodes $w_1$ and $w_4$ of $B_0$.
Observe that the subgraph of $G_k^{(2)}$ consisting of $B_0$ and the two
disjoint 4-cycles
of each copy of the gadget subgraph is a 2-edge~cover of $G_k^{(2)}$
of cost $4k+3$.
Hence, $\cost(\DTWO(G_k^{(2)}))\leq 4k+3$.
Moreover, $\opt(G_k^{(2)}) \geq 7k+3$, see \cite[Proposition~28]{cdgkn:map}.

The graph $G:=G_k^{(3)}$ of the third instance consists of
$k$ copies $L_1,\dots,L_k$ of an \redgadget{} obstruction $L$,
$B_0=w_1,\dots,w_6,w_1$, which is a 6-cycle of cost~three,
and two unit-edges between each $L_i$ and $B_0$.
The subgraph $L$ consists of 8 nodes $v_1,\dots,v_8$ and 11 edges;
there are four zero-edges $v_1v_4,v_2v_3,v_5v_8,v_6v_7$, and
seven unit-edges $v_1v_2,v_1v_5,v_2v_8,v_3v_4,v_4v_6,v_5v_6,v_7v_8$;
see the subgraph induced by the nodes $v_1,\dots,v_8$ in
Figure~\ref{f:S34R8-example}(c); observe that 8~of the 11~edges form two
disjoint 4-cycles (namely, $v_1,v_2,v_3,v_4,v_1$ and $v_5,v_6,v_7,v_8,v_5$)
and the other three edges are $v_1v_5$, $v_2v_8$, and $v_4v_6$.
$G_k^{(3)}$ has two unit-edges between each copy of the gadget subgraph $L_i$
($i=1,\dots,k$) and $B_0$; these two edges are incident to the nodes
$v_4$ and $v_8$ of $L_i$ (see the illustration in
Figure~\ref{f:S34R8-example}(c)) and to the nodes $w_1$ and $w_4$ of $B_0$.
Observe that the subgraph of $G_k^{(3)}$ consisting of $B_0$ and the two
disjoint 4-cycles of
each copy of the gadget subgraph is a 2-edge~cover of $G_k^{(3)}$
of cost $4k+3$.
Hence, $\cost(\DTWO(G_k^{(3)}))\leq 4k+3$.
Moreover, $\opt(G_k^{(3)}) \geq 7k+3$;
this holds because any 2-ECSS of $G$ that contains all the zero-edges
induces a connected subgraph of minimum degree two on the node-set
$V(L)$ of each copy of $L$, and such a subgraph of $L$ has cost
$\geq5$; hence, an optimal 2-ECSS of $G$ contains $\geq5$ of the
unit-edges of $L_i$ as well as the two unit-edges between $L_i$ and
$B_0$, for each $i\in[k]$.

}

\medskip
\noi\textbf{Acknowledgments}: We are grateful to several colleagues for
their careful reading of preliminary drafts and for their comments.

{

\bibliographystyle{abbrv}
\bibliography{map53}

}

\end{document}